\documentclass{article}

\usepackage[margin = 1in]{geometry}
\usepackage{amsmath,amsfonts,amsthm}
\usepackage{xcolor}
\usepackage{hyperref}

\usepackage{url}

\hypersetup{
	colorlinks,
	linkcolor={blue!80!black},
	citecolor={green!80!black},
	urlcolor={blue!80!black}
}

\theoremstyle{theorem}
\newtheorem{theorem}{Theorem}[section]
\newtheorem{claim}[theorem]{Claim}
\newtheorem{lemma}[theorem]{Lemma}
\newtheorem{fact}[theorem]{Fact}

\newtheorem{observation}[theorem]{Observation}
\newtheorem*{lemma*}{Lemma}
\newtheorem*{theorem*}{Theorem}
\newtheorem*{claim*}{Claim}

\theoremstyle{definition}
\newtheorem{definition}[theorem]{Definition}
\newtheorem{remark}[theorem]{Remark}

\DeclareMathOperator{\rank}{\mathrm{rank}}
\DeclareMathOperator{\E}{\mathbb{E}}
\DeclareMathOperator{\pE}{\widetilde{\mathbb{E}}}
\newcommand{\N}{\mathbb{N}}
\newcommand{\Ind}{\mathop{\mathbb{I}}}
\DeclareMathOperator{\Tr}{\mathrm{Tr}}
\DeclareMathOperator{\SA}{\mathrm{SA}}
\DeclareMathOperator{\sgn}{\mathrm{sign}}
\DeclareMathOperator{\Var}{\mathbb{V}}
\DeclareMathOperator{\pV}{\widetilde{\mathbb{V}}}
\DeclareMathOperator{\Cov}{\mathrm{Cor}}

\newcommand{\pC}{\mathop{\widetilde{\mathrm{Cor}}}}

\DeclareMathOperator{\vol}{\mathrm{vol}}
\DeclareMathOperator{\val}{\mathrm{val}}
\newcommand{\argmax}{\mathop{\mathrm{argmax}}}

\newcommand{\Id}{\mathrm{Id}}
\newcommand{\R}{\mathbb{R}}
\newcommand{\cE}{\mathcal{E}}
\newcommand{\cS}{\mathcal{S}}

\newcommand{\iprod}[1]{\langle #1 \rangle}
\newcommand{\hf}{\frac{1}{2}}

\newcommand{\eps}{\varepsilon}

\newcommand{\Tnote}[1]{}
\newcommand{\Snote}[1]{}
\newcommand{\Lnote}[1]{}

\newcommand{\e}{\eps}

\title{Subexponential LPs Approximate Max-Cut}
\author{Samuel B. Hopkins\thanks{\protect\url{hopkins@berkeley.edu} UC Berkeley} \and Tselil Schramm\thanks{\protect\url{tselil@mit.edu} Harvard University and MIT} \and Luca Trevisan\thanks{\protect\url{L.Trevisan@UniBocconi.It} Bocconi University} }

\begin{document}
\maketitle

\begin{abstract}
We show that for every $\varepsilon > 0$, the degree-$n^\varepsilon$ Sherali-Adams linear program (with $\exp(\tilde{O}(n^\varepsilon))$ variables and constraints) approximates the maximum cut problem within a factor of $(\frac{1}{2}+\varepsilon')$, for some $\varepsilon'(\varepsilon) > 0$.
Our result provides a surprising converse to known lower bounds against all linear programming relaxations of Max-Cut \cite{CMM09,KMR}, and hence resolves the extension complexity of approximate Max-Cut for approximation factors close to $\frac{1}{2}$ (up to the function $\varepsilon'(\varepsilon)$).
Previously, only semidefinite programs and spectral methods were known to yield approximation factors better than $\frac 12$ for Max-Cut in time $2^{o(n)}$.
We also show that constant-degree Sherali-Adams linear programs (with $\text{poly}(n)$ variables and constraints) can solve Max-Cut with approximation factor close to $1$ on graphs of small threshold rank: this is the first connection of which we are aware between threshold rank and linear programming-based algorithms.

Our results separate the power of Sherali-Adams
versus Lov\'asz-Schrijver hierarchies for approximating Max-Cut, since it is known \cite{STT} that $(\frac{1}{2}+\varepsilon)$ approximation of Max Cut
requires $\Omega_\varepsilon (n)$ rounds in the Lov\'asz-Schrijver hierarchy.

We also provide a subexponential time approximation for Khot's Unique Games problem \cite{K02}: we show that
for every $\varepsilon > 0$ the degree-$(n^\varepsilon \log q)$ Sherali-Adams linear program distinguishes instances of Unique Games of value $\geq 1-\varepsilon'$
from instances of value $\leq \e'$, for some $\e'( \varepsilon) >0$, where $q$ is the alphabet size.
Such  guarantees are qualitatively similar to those of previous subexponential-time algorithms for Unique Games but our algorithm does not rely on semidefinite programming or subspace enumeration techniques \cite{ABS,BRS11,GS11}.
\end{abstract}

\section{Introduction}

In the Max-Cut problem, we are given an undirected graph $G=(V,E)$ and are asked to find a partition of $V$ into two sets $C$ and $\overline C$ that maximizes the fraction of edges of $E$ having exactly one endpoint in $C$.

Besides  being one of the simplest and most well-studied discrete optimization problems, Max-Cut is held as a hallmark example of the success of semidefinite programs and spectral methods over other algorithms.
The $.878\ldots$ approximation algorithm by Goemans
and Williamson provided the first application of semidefinite programming (SDP) to the design of approximation algorithms with bounded worst-case approximation ratio. While it is easy to achieve a $\frac 12$ approximation in polynomial time (a simple greedy algorithm can find a partition that cuts half the edges in linear time), only semidefinite programming \cite{GW} and spectral methods \cite{T} have been known to achieve approximations better than $\frac 12$ in polynomial time, or even subexponential time.

A fundamental difficulty in breaking the $\frac 12$ barrier for Max-Cut is that an approximation algorithm also provides a {\em certificate}: if we have, say, a $.52$-approximation algorithm, and we run it on a graph in which the Max-Cut optimum is less than the approximation, for example $\leq .51$, then the algorithm will output a cut of value $\leq .51$ and the execution of the algorithm, together with its proof of correctness, will provide a certificate that the Max-Cut optimum is $\leq .51 / .52< .99$. This means that the design of a .52-approximation algorithm requires the development of a technique that is able to provide, for every graph whose Max-Cut optimum is $\leq .51$, a certificate that its Max-Cut optimum is $< .99$. To date, only semidefinite programming was known to provide such certificates in sub-exponential time. Even the certificates implied by the spectral algorithm of \cite{T} are dual feasible solutions of the Goemans-Williamson SDP relaxation.

There has been strong evidence that subexponentially-sized linear programming (LP) relaxations of Max-Cut could not provide such certificates.
Beyond the failure to obtain LP-based algorithms for Max-Cut, there are many concrete strong lower bounds.
Schoenebeck, Tulsiani and Trevisan \cite{STT} prove that in order to achieve an integrality gap better than $\frac 12 +\epsilon$ one needs Lov\'asz-Schrijver hierarchy relaxations \cite{LS91} of size $\exp(\Omega_\epsilon (n))$.
For the more powerful Sherali-Adams hierarchy \cite{SA90}, Charikar, Makarychev and Makarychev \cite{CMM09} prove that to achieve such an integrality gap one needs relaxations of size $\exp(n^{\Omega_\epsilon (1) })$.
In both proofs, the integrality gap instances are (slightly modified) sparse random graphs, which spectral algorithms can solve almost trivially.
Kothari, Meka and Raghavendra \cite{KMR} show that Sherali-Adams integrality gaps for constraint satisfaction problems imply integrality gaps for {\em all} linear programs\footnote{The result
applies to all linear programs obtained in the ``extended formulation'' framework. In the case of Max-Cut, this applies to all relaxations in which
the constraints do not depend on the edges of the graph, and depend only on the number of vertices.} of comparable size,
and, in particular, using the integrality gap of \cite{CMM09}, show that, in order to achieve approximation $\frac 12 + \epsilon$ one needs an LP of size $\exp(n^{\Omega_\epsilon(1)})$.

These results together made any nontrivial LP-based algorithm for Max-Cut seem quite unlikely.
It had been hypothesized that integrality gaps exist for exponential-sized Sherali-Adams LPs, similar to those known for Lov\'asz-Schrijver; the boldest conjecture held that perhaps even random graphs could provide such integrality gaps.
But surprisingly, the latter hypothesis was refuted by O'Donnell-Schramm \cite{ODonnellS19}, who show that the Sherali-Adams hierarchy can certify that a random graph of average degree $d$ has Max-Cut at most $.51$ with an LP of size $\exp(n^{\alpha})$, where $\alpha \rightarrow 0$ as $d \rightarrow \infty$.
This matches the integrality gap result of Charikar et al., up to the precise dependence of $\alpha$ on $d$.

The result of O'Donnell and Schramm can be seen as an average-case analysis of the approximation guarantee of subexponential Sherali-Adams relaxations of Max-Cut, and it applies, more generally, to all graphs whose normalized adjacency matrix has bounded spectral radius.
But even in the wake of the O'Donnell-Schramm result, it was not clear whether subexponential LPs can provide nontrivial approximations for worst-case Max-Cut instances.
Indeed, in the related circumstance of approximating the feasible region of a semidefinite program with $O(n)$ variables and constraints via linear programs, it is known that constant-factor approximation can require $2^{\Omega(n)}$ linear constraints \cite{BFPS}.

In this paper we resolve this question by providing a worst-case analysis that applies to all graphs.
The graphs considered by O'Donnell and Schramm crucially have Max-Cut and Sherali-Adams values close to $1/2$: their analysis relies on fast mixing of random walks, and larger Max-Cuts preclude this.
To overcome this barrier, our key intermediate step (and our main technical contribution) is an analysis of Sherali-Adams relaxations on graphs whose adjacency matrices have both a nontrivial number of large positive eigenvalues and an arbitrary number of large negative ones.
In particular, this includes graphs with Max-Cut values close to $1$.
For this, we distill a simple proof of a strong \emph{local-to-global correlation} lemma (Lemma~\ref{lem:l-t-g}), with significantly weaker assumptions than existing local-to-global lemmas used to analyze LP and SDP hierarchies \cite{ODonnellS19, BRS11, RagT12}.

\subsection{Our Results}

\begin{theorem}[Main Result for Max-Cut] \label{th.main.maxcut}
For every $\alpha > 0$ there is an $\epsilon >0$ such that a Sherali-Adams LP relaxation of Max-Cut of degree $n^\alpha$ provides an approximation ratio at least $\frac 12 + \epsilon$.
\end{theorem}

Up to the precise dependence of $\e$ on $\alpha$ our result is the best possible for Sherali-Adams (and in fact \emph{any} linear program which fits in the \emph{extended formulations} framework) \cite{CMM09,KMR}.

Our approach also offers insight into polynomial-size LPs from the Sherali-Adams hierarchy.
We show that a large class of Max-Cut instances -- graphs with \emph{bounded threshold rank} -- can be solved almost exactly by LPs of polynomial size.
To the best of our knowledge, this is the first known connection between bounded-threshold rank graphs and the guarantees of linear-programming based algorithms; previously such connections were known only for semidefinite programs \cite{BRS11,GS11}.
We prove Theorem~\ref{th.main.maxcut} by combining Theorem~\ref{th.low-threshold-rank} below with a graph decomposition theorem of Steurer~\cite{Steurer10} (closely related to that of Arora, Barak, and Steurer~\cite{ABS}).

\begin{definition}[Threshold Rank]
For $\tau > 0$, we say that a graph $G$ has $\tau$-threshold rank at most $k$ if the normalized adjacency matrix of $G$ has at most $k$ eigenvalues bigger than $\tau$. (Note that we only count the positive eigenvalues larger than $\tau$, not including the number of negative eigenvalues whose absolute value is bigger than $\tau$.)
\end{definition}

\begin{theorem}[Max-Cut on Low Threshold-Rank Graphs]
\label{th.low-threshold-rank}
  Let $G$ be a graph, let $\tau > 0$, and let $r = \rank_\tau(G)$.
  The degree-$k$ Sherali-Adams LP relaxation of Max-Cut on $G$ is an $\e$ additive approximation\footnote{Here we think of the Max-Cut value as normalized to lie in $[0,1]$.}
to the Max-Cut value of $G$ as long as
\[
  k \geq \frac r {\e^{O \left ( \log \tfrac n r / \log \tfrac 1 \tau \right ) } } \, .
\]
\end{theorem}

As an example, Theorem~\ref{th.low-threshold-rank} says that $O(r)$ levels of the Sherali-Adams hierarchy approximate Max-Cut up to $0.01$ additive error on graphs whose normalized adjacency matrices have $r$ eigenvalues larger than $n^{-0.01}$.
Theorem~\ref{th.low-threshold-rank} captures the results of \cite{ODonnellS19} as a special case, but it also extends to graphs with nontrivial eigenvalues (both positive and negative): this is crucial to our ability to use it in the proof of Theorem~\ref{th.main.maxcut}.

Our approach also extends beyond Max-Cut.
For a variety of $2$-CSPs, including Max-$2$-Lin and Max-$k$-Cut, we show that Sherali-Adams LPs of subexponential size obtain nontrivial worst-case approximations.
Of particular note is our result for Unique Games, where we show that $O( n^\alpha\log q)$-degree Sherali-Adams relaxations provide a constant-factor approximation for Unique Games on alphabets of size $q$.\footnote{Recall that the Unique Games Conjecture says that such approximations are NP hard.}
This is qualitatively similar in performance to the spectral algorithm of Arora et al. \cite{ABS} and the semidefinite programming algorithm of Barak et al.~\cite{BRS11}.

\begin{theorem}[Main Result for Unique Games] \label{th.main.ug}
  For every $\alpha > 0$ there is an $\e > 0$ such that a degree-$n^\alpha\log q$ Sherali-Adams relaxation of Unique-Games on an alphabet of size $q$ distinguishes instances of value $\leq \e$ from instances of value $\geq 1-\e$.
\end{theorem}

The spirit of these results is that subexponentially-sized LPs can (surprisingly) match the performance of SDPs for some key problems in combinatorial optimization.
In the case of Unique Games in the sub-exponential time regime, the performance of the LP is qualitatively similar to what is known for SDPs (which are the best known algorithms) -- both achieve constant-factor approximations in subexponential time.\footnote{The precise approximation factors obtained by the algorithms of \cite{ABS,BRS11} are better, but the bottom line is that we obtain an approximation ratio independent of the alphabet size in subexponential time, which is the measure of success relevant to the Unique Games Conjecture.}
In the polynomial-size regime, Theorem~\ref{th.low-threshold-rank} shows that Sherali-Adams LPs match the best known performance guarantees for SDP-based hierarchies for Max-Cut on graphs of low threshold rank, with the exception of a somewhat more stringent requirement on the value of the threshold $\tau$.\footnote{For LPs, our results are meaningful mainly when $\tau \leq n^{-\Omega(1)}$, while SDP-based hierarchies require only $\tau \leq \text{poly}(\epsilon)$. The lower bounds of \cite{CMM09} show that this $n$-dependence of the threshold is necessary for LPs; our result matches their lower bound in the case of expander graphs, i.e. for threshold rank $1$.}

One might wonder how universal this phenomenon is, and in particular whether it extends beyond constraint satisfaction problems.
We observe that the approximation factor achieved by Sherali-Adams linear programs for the Max-QP problem (also known as the Max-Cut-Gain problem with negative edge weights) is exponentially worse than what is achieved by polynomial-size semidefinite programs, even in the subexponential regime.

\begin{observation}\label{obs.main.qp}
Degree $k$ Sherali-Adams relaxations provide a $\Theta(n/k)$-approximation to Max-QP.
\end{observation}

\noindent Despite the simplicity of the proof of Observtion~\ref{obs.main.qp}, we are not aware of a similar statement in the literature.
 Charikar and Wirth show that a basic (polynomial-size) semidefinite program provides an $O(\log n)$ approximation to this problem \cite{CW04}.
Thus, we have a simple example of an optimization problem for which there is a wide gap between SDP and Sherali-Adams performance even in the subexponential regime.

\subsection{Overview of the proof of our results}

\paragraph{The Sherali-Adams hierarchy.}
A feasible solution of a degree-$r$ Sherali-Adams Max-Cut LP with value $c$ describes a relaxation of a distribution over cuts that, on average, cut a $c$ fraction of the edges. 
The solution contains a complete description of the marginals of such a distribution over subsets of at most $r$ vertices.
 These marginal distributions are locally consistent with one another, in the sense that the distribution for a set $A$ of vertices and the distribution of a set $B$ have the same marginal distribution on $A\cap B$. 

Because the LP is a relaxation, there may in fact be no actual distribution over cuts which is consistent with the marginal distributions described by the feasible solution.
For example, if the graph underlying the Max-Cut instance is a clique, then if we represent cuts as $\pm 1$ assignments to vertices, a feasible degree-2 solution of value 1 is to define, for every pair $u,v$ of vertices, the distribution that sets $X_u=-1, X_v=1$ with probability $1/2$ and $X_u=1,X_v=-1$ with probability $1/2$. These local distributions agree on their intersections, because the local distributions on vertex signs $X_u,X_v$ and on $X_v,X_w$ agree on their marginal distribution on $X_v$ (which, in both cases, is equally likely to be $-1$ or $1$).
On the other hand, the max cut value is $\frac{1}{2} +o(1)$.

\paragraph{Local correlation, global correlation, and independent rounding.}
We now describe some of the ideas that we build on from prior work on rounding LP or SDP relaxations.

A tempting approach to round a Sherali-Adams Max-Cut relaxation is to assign each vertex randomly according to the local distribution for that vertex, treating each vertex independently. 
This {\em independent rounding} approach fails if, for a typical edge $(u,v)$, the local joint distribution for the pair of vertex signs $\{ X_u, X_v \}$ differs noticeably from the product distribution $\{X_u\}\{X_v\}$ in which $X_u$ and $X_v$ are assigned independently. 
In the example of the degree-2 solution for the clique that we discussed above, each edge $u,v$ is cut with probability $1/2$ by independent rounding but with probability $1$ by the local distribution on $X_u,X_v$.

If, however, for a typical edge $(u,v)$ the local distribution on $X_u,X_v$ is close (for example, in $\ell_1$ norm) to the product distribution on the marginals (that is, if $X_u$ and $X_v$ have ``low correlation'' in their local distribution), then independent rounding works. Thus, we would like to take an arbitrary solution and reduce it to a solution in which the correlation of $X_u,X_v$ according to their local distributions is small on average for a random edge $(u,v)$. Such a solution is said to have low {\em local correlation}.
It is not difficult to show that if a solution has value $c$, and the average, for a random edge $(u,v)$, of the $\ell_1$ distance between the local distribution on $X_u,X_v$ and the product distribution on the marginals is at most $\E_{u \sim v} \|\{X_u,X_v\}-\{X_u\}\{X_v\}\|_1 \le
\epsilon$, then independent rounding will cut at least a  $c-\epsilon$ fraction of edges (in expectation).

The {\em global correlation} of a feasible solution is the average, over all pairs $u,v$ of the $\ell_1$ distance: $\E_{u,v} \|\{X_u,X_v\}-\{X_u\}\{X_v\}\|_1$.
The global correlation of a feasible solution can be reduced using an operation called {\em conditioning}.
Starting from a feasible solution of degree $r$, conditioning over $t$ variables yields a new feasible solution of degree $r-t$, of the same value, and 
such that the global correlation is at most $O(1/\sqrt{t})$. (See Theorem \ref {thm:gcr-conditioning} for a more precise statement.)

Together these facts imply that, if we could argue that low global correlation implies low local correlation, then we would have a good rounding procedure:  first apply  conditioning (to reduce global, and therefore local, correlation) and then apply independent rounding. 
This scheme underlies several approximation algorithms using LP and SDP hierarchies; it was pioneered in \cite{BRS11,RagT12} and then used in a number of subsequent works.

\paragraph{Local to global correlation in low threshold rank graphs.}
Our first main result (Lemma \ref{lem:loc-to-t}) says the following: suppose that, for a random edge $(u,v)$ of $G$ the correlation (as defined above) between $X_u$ and $X_v$ is at least $\gamma$. Then, if we pick a random walk $v_0, v_1, \ldots, v_t$  of length $t$ from a random start vertex $v_0$, the average correlation between $X_{v_0}$ and $X_{v_t}$ is $\gamma^{\Omega(t)}$, provided that the degree of the Sherali-Adams relaxation
is  at least $O(1/\gamma)^t$.

Next, we argue that in graphs of {\em bounded threshold rank}, noticeably large correlation along sufficiently long random walks implies noticeably large global correlation.

For a parameter $\tau > 0$, we say that a graph has $\tau$-threshold rank at most $k$ if the normalized adjacency matrix of the graph has at most $k$ eigenvalues bigger than $\tau$. (Note that we only count the positive eigenvalues larger than $\tau$, and we do not count the number of negative eigenvalues whose absolute value is bigger than $\tau$.) Graphs of bounded threshold rank have been studied before from the point of view of the performance of SDPs on such graphs, but not, as far as we know, from the point of view of the performance of LPs.

The key property of graphs of $\tau$-threshold rank at most $k$ is this: for a typical start vertex, the distribution of the last vertex
of a random walk of length $t = O\left( \frac 1 {\log 1/\tau} \log n\right)$ has collision probability $O(k/n)$ (see Claim  \ref{claim:thresh-trace}).
In particular, if the expected correlation between the endpoints
of a $t$-step random walk is at least $\delta$, then the global correlation is at least $\Omega(\delta/k)$.

In summary, we have that: (i)  an upper bound on the correlation along $t$-steps random walks implies an upper bound on the local correlation; (ii) in graphs of bounded threshold rank, an upper bound on the global correlation implies an upper bound on the correlation along random walks; (iii) one can get a feasible solution of bounded global correlation by applying conditioning. Putting these facts together, if we have a graph of $\tau$-threshold rank $k$, and we have a Sherali-Adams solution, conditioning 
on $O(k^2 n ^{O( (\log (1/\gamma)) / \log (1/\tau))})$ variables (which we can do if the degree of the Sherali-Adams
solution is larger than the above bound), gives us a solution of local correlation at most $\gamma$.
The number of variables that we need to condition on is at most the desired bound $n^\alpha$ if, say,
$k < n^{\alpha/4}$ and $\tau$ is sufficiently small relative to $\gamma$ and $\alpha$.

\Lnote{Added this comparison to OS}
In comparing this approximation to the result of  O'Donnell-Schramm \cite{ODonnellS19}, we
see that they can guarantee a $1-\gamma$ approximation using a Sherali-Adams relaxation of degree $n^\alpha$,
provided that all the non-trivial eigenvalues of the normalized adjacency matrix are at most $\tau$ in magnitude, with
similar tradeoffs between $\gamma$, $\alpha$ and $\tau$ as we have above.
The requirement that all the non-trivial eigenvalues are at most $\tau$ in magnitude precludes graphs with large cuts (say, Max-Cut value $0.99$), meaning that \cite{ODonnellS19} cannot analyze the Sherali-Adams relaxation for all graphs.
In the analysis sketched
above, we can deal with graphs that have up to $n^{\alpha/4}$ positive eigenvalues larger than $\tau$, instead of just one, and arbitrarily many, even order of $n$, negative eigenvalues of magnitude larger than $\tau$, instead of none. The fact that we do not put any restriction on the number of large-in-magnitude negative eigenvalues enables us to extend our analysis to general graphs by partitioning general graphs into graphs of bounded threshold rank.

\paragraph{Partitioning into pieces of bounded threshold rank.}
Finally, we reduce the case of general graphs to the case of graphs of bounded threshold rank via a decomposition theorem (Theorem \ref{thm:partition}) that shows that every graph, after the removal of a bounded number of edges, breaks down into connected components each of bounded threshold rank. We employ such a decomposition theorem proved originally in~\cite{Steurer10}, which is closely related to a graph decomposition theorem proved Arora, Barak and Steurer~\cite{ABS}. The latter concerns thresholds $\tau$ close to $1$, while
we are interested in $\tau$ close to zero, because
eventually our running time will be exponential in $n^{O(1 / \log 1/\tau)}$.
Steurer \cite{Steurer10} shows that graphs can be decomposed into pieces of small $\tau$-threshold rank for $\tau$ close to $0$, although one has to remove a large fraction of edges to achieve it. 

\paragraph{Obtaining a solution.}
We first remove edges to make the residual graph have connected components of bounded threshold rank; under the assumption that the value of the Sherali-Adams solution is sufficiently high to begin with, the resulting components will have Sherali-Adams value close to $1$. 
Then we apply conditioning and rounding independently in each component, obtaining a cut of value close to $1$ in each component.
 Interpreting each cut as a $\pm 1$ assignment to the vertices, we then, independently for
each component, either leave the cut as is with probability $\frac{1}{2}$ or ``flip'' the cut (switching $-1$s with $1$s) with probability $\frac{1}{2}$. This last step ensures that
all the edges between components are cut with probability $\frac{1}{2}$. 
Thus our cut has value at least $\frac{1}{2} + \varepsilon$ for $\varepsilon$ proportional to the fraction of edges not cut by the partition.

\paragraph{Unique Games.}
Our result for Unique Games has a similar structure: we define a measure of correlation and prove that noticeably large
local correlation implies noticeably large correlation between the endpoints of random walks, which implies noticeably large
global correlation in graphs of bounded threshold rank. The main difference compared to the Max-Cut analysis is in the measure of correlation between the random variables corresponding to two vertices. We use a notion that we call ``permutation correlation,'' which has previously been used as a measure of correlation in the analysis of SDPs for Unique Games, and we are able to show that if the local permutation correlation along a random edge is $\gamma$, then the permutation correlation along the endpoint of a  $t$-step random walk is at least $\gamma^{O(t)}$, where the constant in the big-Oh is an absolute constant that does not depend on the size of the alphabet of the unique game.
Crucially, in the case of Unique Games, rounding a Sherali-Adams solution by indpendently sampling from the distributions $\{X_u\}$ succeeds if the solution only has small local permutation correlation (a weaker requirement than small local correlation).

\subsection{Discussion}

Our analysis of Sherali-Adams linear programs is similar to the rounding of SDPs in \cite{BRS11}: we partition the graph into components of bounded threshold rank, we show that noticeably large local correlation implies noticeably large global correlation in such graphs, and we observe that large global correlation is a property that can survive only a bounded number of conditionings.
However, in the case of SDPs, the spectrahedral constraints allow one to prove a local-to-global correlation lemma directly leveraging spectral properties of the underlying constraint graph in a direct way.
For linear programs, it is surprising that we are able to take advantage of spectral properties at all.\footnote{Though it is less surprising in the wake of the results of O'Donnell and Schramm \cite{ODonnellS19}.}

The main step is Lemma \ref{lem:loc-to-t}, which allows us to relate the local correlation along one edge to the correlation between the first and the last vertices of a random walk, and hence allows us to relate local correlation to global correlation in graphs of bounded threshold rank. 

The rounding algorithm that we use in this paper could be applied to a feasible Lov\'asz-Schrijver solution:
a Lov\'asz-Schrijver solution can be rounded near-optimally if it exhibits low local correlation. 
The operation of conditioning is well-defined for Lov\'asz-Schrijver solutions,
and conditioning on $t$ variables reduces the global correlation to $O(1/\sqrt t)$ at the cost of reducing the degree of the solution by $t$. 
Our algorithm, however, must fail, because of the known integrality gaps, and what fails is Lemma \ref{lem:loc-to-t}.

It would be interesting to see a similar phenomenon in the SDP setting: could there be a rounding algorithm for a Sum-of-Squares (SoS) relaxation that would be well-defined on weaker relaxations, but whose analysis relies on properties that hold only for SoS relaxations? 
For example, it is an open problem whether SoS relaxations of polynomial size can refute the Unique Games Conjecture, and it is known that weaker relaxations of polynomial (or even slightly superpolynomial) size cannot refute the Unique Games Conjecture \cite{RS09}. 
This is often interpreted as evidence that, in order to refute the Unique Games Conjecture via SoS relaxations, one would have to develop a radically new rounding technique that makes explicit use of SoS constraints. 
Could it be, instead, that there is a relatively simple rounding scheme for SoS relaxations of Unique Games, that is well defined on weaker relaxations and such that SoS-specific constraints are used only in the analysis?

Finally, we note that our subexponential algorithm for $(1-\e)$ vs $\e$ Unique Games (that is, distinguishing a $1-\e$-satisfiable instance from an $\e$-satisfiable one) requires $n^{\e'(\e)} \log q$ rounds of Sherali-Adams, while the prior algorithm based on hierarchies required $n^{\e'(\e)} q^3$ rounds.
As far as we know, this is the first subexponential algorithm allowing $q \geq n^{\Omega(1)}$.
However, our result is not a strict improvement on the prior algorithm \cite{BRS11} (which uses a stronger hierarchy), because that algorithm solves $(1-\e)$ vs $1/2$ Unique Games.


\paragraph{Open Problems}
We mention a few open problems; resolving any of these (affirmatively or negatively) would be of interest.

\begin{enumerate}
  \item \emph{Beating a random assignment for general $2$-CSPs?} Hastad \cite{H08} shows that for every $2$-CSP with alphabet size $q$, an SDP obtains an approximation ratio that is strictly better than the ratio obtained by randomly sampling assignment in $[q]^n$.
  Our work shows that subexponential LPs offer similar guarantees for some $2$-CSPs including Max-Cut and Unique Games.
  Do subexponentially-sized linear programs offer a nontrivial approximation for every $2$-CSP?

  \item \emph{Refined approximations for Max-Cut?} In addition to providing a $0.878\ldots$ approximation to Max-Cut, the Goemans-Williamson SDP also offers more refined guarantees.
  (1) In a graph with Max-Cut value at least $1-\epsilon$ it finds a cut of size $1-O(\sqrt{\epsilon})$ \cite{GW}, and (2) via an alternative rounding scheme due to Charikar and Wirth \cite{CW04}, in a graph with Max-Cut value $\tfrac 1 2 + \epsilon$ it finds a cut of size $\tfrac 1 2 + \Omega(\e / \log(1/\e))$.
  Can these guarantees be matched by subexponentially-sized linear programs?
  Our analysis cannot be extended to these settings as-is because the graph partitioning scheme we employ forces us to settle for randomly cutting most of the edges of the graph (all but, say, a $0.01$ fraction), which would seem to preclude ever finding a cut of size close to $0.99$.


  \item \emph{LP Certificates of Expansion?}
  Providing certificates of graph expansion is another combinatorial optimization problem for which spectral methods and SDPs appear to out-perform LPs.
  For instance, in every graph with expansion $0.99$, Cheeger's inequality says that the second eigenvalue of the adjacency matrix certifies that the graph has expansion at least $0.01$, while the best similar result for LPs loses a factor of $\log n$ \cite{LR99}.
  Can subexponential LPs offer certificates comparable to the second eigenvalue?
\end{enumerate}

\subsection*{Organization}
The definitions we need to work with CSPs and the Sherali-Adams hierarchy are in Section~\ref{sec:prelims}.
In Section~\ref{sec:max-cut} we prove Theorem~\ref{th.main.maxcut} and Theorem~\ref{th.low-threshold-rank} on Max-Cut.
In Section~\ref{sec:ug} we prove Theorem~\ref{th.main.ug} on Unique Games and describe a modification of the arguments to apply to a broader class of permutation-symmetric CSPs such as Max-$2$-Lin and Max-$k$-Cut.
In Appendix~\ref{sec:qp} we prove Observation~\ref{obs.main.qp}.
The remaining appendices contain reproductions of proofs from prior work for completeness (on spectral partitioning and global correlation rounding).

\section{Preliminaries}
\label{sec:prelims}
In this section we provide some preliminaries; the reader may wish to proceed directy to the proof of our main theorem in Section~\ref{sec:max-cut}.
For guidance, this section  is organized as follows:
Section~\ref{sec:graphs} introduces some definitions and notation regarding graphs. Section~\ref{sec:cor} introduces measures of correlation that we will use to track the local-vs.-global correlation in our pseudodistributions.
Section~\ref{sec:2csps} introduces constraint satisfaction problems, and defines Max Cut and Unique Games.
Finally, Section~\ref{sec:SA} introduces the Sherali-Adams LP and the notion of local random variables.

\subsection{Graphs}\label{sec:graphs}
A graph $G$ on $n$ vertices is a collection of nonnegative \emph{weights} $w_{ij} \geq 0$ for each pair $\{ij\} \in {n \choose 2}$.
In this work all graphs are simple, undirected, and contain no isolated vertices, but may be irregular.

\begin{definition}[Adjacency, Normalized Adjacency, and Walk Matrices]
  For a graph $G$ we often employ the \emph{adjacency matrix} $A$ with entries $A_{ij} = w_{ij}$, the \emph{degree matrix} with entries $D_{ii} = \sum_{j} w_{ij}$, the \emph{walk matrix} $D^{-1} A$, and the \emph{normalized adjacency matrix} $N = D^{-1/2} A D^{-1/2}$.
\end{definition}

We will use the following notation:

\begin{definition}[Random Walk Endpoints]
  For a graph $G$ on $n$ nodes, we often use $i \sim_\ell j$ to denote the distribution pairs $i,j$ where $i$ is first chosen according to the stationary measure $\pi$ of the random walk on $G$ and then $j$ is the result of an $\ell$-step random walk initialized at $i$.
For the special case $\ell = 1$, we will use the shorthand $i \sim j$.
\end{definition}

\begin{definition}[Threshold Rank]
A graph $G$ has {\em $\tau$-threshold rank $k$}, or $\rank_{\tau}(G) = k$, if the normalized adjacency matrix of $G$ has at most $k$ eigenvalues larger than $\tau$.
\end{definition}

\subsection{Random Variables and Measures of Correlation}\label{sec:cor}

For a random variable $X$ we denote by $\{X\}$ the associated density function.

\begin{definition}[Correlation]
If $X,Y$ are jointly distributed taking values in $[q] \times [q]$, we will often be interested in the following notion of correlation between $X,Y$:
\[
\Cov(X,Y) = \sum_{a,b \in [q]} |\Pr(X=a,Y=b) - \Pr(X=a)\Pr(y=b)| \, .
\]
Notice that this is exactly the $\ell_1$ distance $\|\{X,Y\} - \{X\} \times \{Y\} \|_1$.
Pinsker's inequality shows that $\Cov(X,Y) \leq \sqrt{2 I(X;Y)}$, where $I(\cdot ; \cdot)$ denotes the mutual information between $X$ and $Y$.
\end{definition}

\begin{definition}[Permutation Correlation]\label{def:pcor}
We will also use a related notion of correlation when proving our results for Unique Games.
Rather than measuring the correlation across all outcomes in the joint distribution over two variables, we measure only the correlation along the permutation from $[q]$ to $[q]$ that maximizes the difference between the joint distribution and the product of the marginal distributions.
For variables $X,Y \in [q]$, define
\[
\Cov_\pi(X,Y) = \max_{\pi \in \cS_q} \sum_{a \in [q]} \left|\Pr(X = a, Y = \pi(a))] - \Pr(X = a)\Pr(Y= \pi(a))\right|.
\]
We will refer to this as the ``permutation correlation''.
This notion was introduced in prior work on unique games on expanders \cite{AroraKKSTV08}, and later used as well in \cite{BRS11}.
\end{definition}
We note that since $\Cov_\pi(X,Y) \le \Cov(X,Y)$, as before we can relate $\Cov_\pi(X,Y) \le \sqrt{2I(X;Y)}$ via Pinsker's inequality.

\begin{definition}[Variance for Discrete Random Variables]
We also introduce a notion of variance for a $[q]$-valued random variable $X$.
For each $a \in [q]$, let $X_a$ be the $0/1$ variable such that $X_a = \Ind(X = a)$.
Then we let $\Var(X) = \sum_{a \in [q]} \Var(X_a)$.
We observe that since $\Var(X_a) \leq \E X_a$, we have $\Var(X) \leq \sum_{a \in [q]} \E X_a = 1$.
\end{definition}

\subsection{2CSPs}\label{sec:2csps}
An $n$-variable instance of 2CSP with alphabet size $q \in \N$ consists of a pair $(G,\Pi)$, where $G$ is an $n$-node weighted graph with weights $w_{ij}$ having $\sum_{ij} w_{ij} = 1$ and $\Pi$ is a collection of functions $\Pi_{ij} \, : \, [q] \times [q] \rightarrow \{0,1\}$ for every edge in $G$.
Without loss of generality throughout the paper we assume $w_{ij} \in [w_{\max}/n^3, w_{\max}]$, since low-weight edges of weight much less than $w_{\max} / n^2$ can be thrown out.

\begin{definition}[Objective Value]
The objective value of an assignment $x \in [q]^n$ for an instance is $\sum_{i \sim j} w_{ij} \Pi_{ij}(x_i,x_j) = \E_{i \sim j} \Pi(x_i,x_j)$, where $w_{ij}$ is the weight of edge $i,j$.
\end{definition}

\begin{definition}[Max-Cut]
An instance of Max-Cut is an instance of 2CSP where $q = 2$ and $\Pi_{ij}(x,x') = 1$ if and only if $x \neq x'$.
\end{definition}

\begin{definition}[Unique Games]
An instance of Unique Games is an instance of 2CSP where $\Pi_{ij}(x,x')$ represents a bijective map from $[q]$ to $[q]$.
\end{definition}

\subsection{Local Distributions and Sherali-Adams Linear Programs}\label{sec:SA}
We briefly discuss basic definitions involving the Sherali-Adams linear programming hierarchy.
For much more detail and proofs, see \cite{FKP19}.

\begin{definition}[Local pseudodistribution]
For $q,n \in \N$ and $t \leq n$, a $q$-ary $t$-local pseudodistribution is a collection $\mu = \{\mu_S\}_{S \in {n \choose t}}$ of probability distributions $\mu_S$ on $[q]^t$ such that for every $S' \subseteq [n]$ with $|S'| \leq t$, the marginal distributions of $\mu_{S}$ for $S \supseteq S'$ on the variables in $S'$ are all identical.

We often abuse notation and instead write $X_1,\ldots,X_n$ as {\em $t$-local random variables} induced by $\{\mu_S\}$, with the understanding that only probabilities and events concerning fewer than $t$ variables at once are well defined.
In particular, for two $t$-local random variables $X_i,X_j$, when $t \geq 2$ their correlation $\Cov(X_i,X_j)$ is well defined.
We often write $\pC(X_i,X_j)$ as a reminder that the underlying variables are only $t$-local.
\end{definition}

\begin{definition}[Sherali-Adams Polytope]
The set of all $q$-ary $t$-local pseudodistributions on $n$ variables is the \emph{degree-$t$ Sherali-Adams polytope} -- it is standard (see \cite{FKP19}) that this is a polytope involving $(qn)^{O(t)}$ variables and constraints.
\end{definition}

\begin{definition}[Pseudoexpectation]
The $t$-local pseudodistributions are in one-to-one correspondence with \emph{Sherali-Adams pseudoexpectations}, which are linear maps $\pE \, : \, \R[x_{ia}]^{\leq t}_{i \in [n], a \in [q]} \rightarrow \R$ which satisfy $\pE 1 = 1$ and $\pE f(x) \geq 0$ if $f$ is a nonnegative function depending on $\{x_{ja}\}_{j \in S, a \in [q]}$ for some $S \subseteq [n]$ with $|S| \leq t$.
(Here $\R[y]^{\leq t}$ denotes polynomials in variables $y$ with real coefficients and degree at most $t$.)

Again abusing notation, we often call such a pseudoexpectation $\pE$ or the corresponding pseudodistribution $\mu$ (alternatively written $X_1,\ldots,X_n$) a \emph{Sherali-Adams pseudodistribution}.
If $q=2$, we call it a \emph{Boolean} Sherali-Adams pseudodistribution.
\end{definition}

\begin{definition}[Objective Value of Local Distribution for 2CSP]
If $(G,\Pi)$ is a $q$-ary $2$CSP instance and $\{\mu_S\}$ is $t$-local pseudodistribution for $t \geq 2$, the objective value of $\{\mu_S\}$ for $(G,\Pi)$ is given by
\[
\sum_{i \sim j} w_{ij} \cdot \Pr_{X_i,X_j \sim \mu_{ij}}( \Pi_{ij}(X_i,X_j) = 1)
\]
where $\mu_{ij}$ denotes the marginal distribution of any $\mu_S$ for $S \supseteq \{i,j\}$ on the indices $i,j$.
If $(G,\Pi)$ is a $q$-ary $2$CSP, we write $SA_t(G,\Pi)$ for the maximum objective value achieved by any $q$-ary $t$-local distribution.
\end{definition}

\section{Subexponential Linear Programs for Max-Cut}\label{sec:max-cut}
In this section we prove our main result on approximating Max-Cut by subexponential-size LPs from the Sherali-Adams hierarchy.
We begin with an overview stating the three main ingredients of the proof.

\paragraph{Local-to-global correlation}
First we prove a local-to-global lemma in graphs of low threshold rank.
Our first step is to use the ``spider random walk'' technique of O'Donnell and Schramm to establish that nontrivial correlation across edges implies nontrivial correlation for the endpoints of random walks of length $t$:
\begin{lemma}\label{lem:loc-to-t}
Let $G$ be a graph on $n$ vertices, let $t$ be a power of two, and let $X_1,\ldots,X_n$ be $(2\lceil(\frac{1}{\gamma})^t\rceil + 1)$-local Boolean random variables.
Then $\E_{i \sim j} \pC(X_i,X_j) \ge 16\gamma$ implies that $\E_{i \sim_t j} \pC(X_i,X_j) \ge \gamma^t$.
\end{lemma}

We then use Lemma~\ref{lem:loc-to-t} to prove a second statement that allows us to relate the correlation of endpoints of long random walks in the graph to the correlation of a uniformly random pair of vertices.
The bound on the correlation crucially depends on the threshold rank of the graph (rather than on the second eigenvalue as a proxy for the mixing time, as is the case in O'Donnell-Schramm).
\begin{lemma}\label{lem:l-t-g}
If $G$ is a graph on $n$ vertices with $\rank_{\tau}(G) \le k$, $t$ is a power of two, and $X_1,\ldots,X_n$ are $2(\frac{1}{\gamma})^{t}$-local random variables, then an upper bound on the global squared correlation 
\[
\E_{i,j\sim \pi} \pC(X_i,X_j)^2 \le \frac{1}{2} \frac{\gamma^{2t}}{k + n\tau^{2t-1}}
\]
for $\pi$ the stationary measure of $G$ implies an upper bound on the local correlation 
\[
\E_{i \sim j} \pC(X_i,X_j) \le 16\gamma.
\]
\end{lemma}
This lemma is a composition of Lemma~\ref{lem:loc-to-glob} and Claim \ref{claim:thresh-trace}, which we prove in Sections~\ref{sec:loc-to-glo}~and~\ref{sec:trace-bd} below.

\paragraph{Global correlation rounding}
Once we establish a sufficiently strong relationship between local and global correlation in low-threshold rank graphs, we can apply the global correlation rounding technique pioneered by \cite{BRS11}.
We will use the following facts which originate in \cite{BRS11,RagT12} and are by now standard.
The first fact says that in expectation, {\em global} correlation drops under conditioning, while the objective value remains the same.
\begin{theorem} \label{thm:gcr-conditioning}
Let $q,n \in \N$, let $\kappa > 0$, and let $\mu$ be a $(\tfrac {6\log q}{\kappa} + 2)$-local pseudodistribution over $[q]$-valued random variables $X_1,\ldots,X_n$.
Let $\pi \in \Delta_n$ be a distribution on $[n]$.
Let $i_1,\ldots,i_k \in [n]$ be indices chosen i.i.d. from $\pi$, with $k = \tfrac{6\log q}{\kappa}$.
There exists $t \leq k$ such that
\[
\E_{i_1,\ldots,i_t} \E_{i,j \sim \pi} \E_{X_{i_1},\ldots,X_{i_t}} \pC(X_i, X_j | X_{i_1},\ldots,X_{i_{t}})^2 \leq \kappa \, .
\]
Furthermore, if $X_1,\ldots,X_n$ are variables in a 2CSP instance $(G,\Pi)$, then
\[
\E_{i_1,\ldots,i_k} \left(\E_{i\sim j} \E_{X_1,\ldots,X_k,X_i,X_j} [\Pi(X_i,X_j)~|~X_1,\ldots,X_k]\right) = \E_{i\sim j} \E_{X_i,X_j}[\Pi(X_i,X_j)].
\]
Note that the expectations taken above over $X_i$ are well-defined for local random variables because each depends on at most $k+2$ variables.
\end{theorem}

The second fact states that in 2CSP instances with low {\em local} correlation, independent rounding produces a solution with high objective value.

\begin{lemma}[Rounding low-correlation Sherali-Adams]\label{lem:gcr-rounding}
Let $(G,\Pi)$ be an instance of 2CSP and let $X_1,\ldots,X_n$ be $q$-ary $T$-local random variables for $T \geq 2$.
Suppose that the local correlation $\E_{i\sim j} \pC(X_i,X_j) \le \delta$.
Let $Y_i \sim \{X_i\}$ be independent samples from the $1$-wise marginals of $X_1,\ldots,X_n$.
Then $\E_Y \E_{i \sim j} \Pi(Y_i,Y_j) \geq \E_{i \sim j} \E_{X_i,X_j} \Pi(X_i,X_j) - \delta$.
Furthermore, this rounding scheme can be derandomized in polynomial time.
\end{lemma}

We will prove both statements in Appendix \ref{app:gcr} for completeness.

\paragraph{Graph partitioning}
In a graph with high threshold rank, we will perform partitioning into low threshold rank parts in the style of \cite{ABS}.
The partitioning scheme of \cite{ABS} partitions a graph into parts of expansion at most $\epsilon$ when $\rank_{1-\epsilon'}(G)$ is large, for $\epsilon,\epsilon'$ close to $0$.
However, their result does not give guarantees for graphs that have large $\tau$-threshold rank when $\tau$ is close to $0$ rather than $1$.
A modification of the \cite{ABS} partitioning scheme for the small-$\tau$ regime appears in \cite{Steurer10}, which shows that if $\rank_{\epsilon}(G)$ is large, then one can obtain a partition of expansion at most $1 - \epsilon'$.

\begin{theorem}[Restatement of \cite{Steurer10} Theorem 2.2]\label{thm:partition}
Fix any $\tau,\alpha \in (0,1)$, and take $n$ sufficiently large.
Any $n$-vertex simple graph $G =(V,E)$ admits a partition into components $G_1,\ldots,G_m$ such that for all $i \in [m]$, the threshold rank $\rank_{\tau}(G_i) \le n^{\alpha}$, with the total fraction of edges cut in the partition bounded by $\frac{1}{|E|}|\cup_{i\neq j\in[m]} E[G_i,G_j]| \le 1 - \exp\left(-O\left(\frac{1}{\alpha^2} \log \frac{1}{\tau}\right)\right)$.
Furthermore, there is a polynomial-time algorithm to compute this partition.
\end{theorem}
We provide a proof in Appendix~\ref{app:partition} for completeness.

\paragraph{Putting things together}
With these three pieces in place, we can prove Theorem~\ref{th.main.maxcut}: we apply the partitioning theorem (Theorem~\ref{thm:partition}) to partition our graph into pieces of small threshold rank, cutting at most a $1-\eps$ fraction of edges in the partition.
Then, within each piece, we apply global correlation rounding (Theorem~\ref{thm:gcr-conditioning}) until the global correlation is low.
From our local-to-global lemma (Lemma~\ref{lem:l-t-g}), we can conclude that the local correlation within each piece is small. 
Furthermore, under the assumption that the objective value is $\ge 1-\eps'$ for $\eps' \ll \eps$, on average the objective value within the parts will be large and thus from Lemma~\ref{lem:gcr-rounding} independent rounding will return a solution satisfying (say) a $\ge .75$-fraction of edges within each piece, which (when aggregated across the parts) gives objective value $\ge 0.75 \eps$.
Finally, by applying a random sign change to the solution within each piece, the $(1-\eps)$ fraction of edges crossing the partition are cut with probability $\frac{1}{2}$, for a solution of objecive value $\frac{1}{2} + .25 \eps$. 
We give a formal version of this argument in Section~\ref{sec:together}.

\subsection{Local-to-Global Lemma for Sherali--Adams}
\label{sec:loc-to-glo}
In this section, we prove a lemma in the style of O'Donnell-Schramm \cite{ODonnellS19} that allows us to relate local correlations to the correlations of independently sampled vertices.
One key difference between our proof and that of \cite{ODonnellS19} is that we track the distance (in $\ell_1$) between joint distribution $\{X,Y\}$ and the product of marginals $\{X\}\{Y\}$, rather than the correlation $\E XY$.\footnote{
It turns out that this is crucial for use with global correlation rounding, as the correlation is on average unaffected by conditioning.}
Another difference is that rather than measuring the distance of the random walk to mixing in terms of the second eigenvalue, we measure it in terms of a trace of a power of the random walk matrix, which allows us to take advantage of graphs of low threshold rank.

\begin{lemma*}[Restatement of Lemma~\ref{lem:loc-to-t}]
Let $G$ be a graph on $n$ vertices, let $t = 2^\ell$ be a power of two, and let $X_1,\ldots,X_n$ be $(2\lceil(\frac{16}{\gamma})^t\rceil + 1)$-local Boolean variables.
Then $\E_{i \sim j} \pC(X_i,X_j) \ge \gamma$ implies that $\E_{i \sim_t j} \pC(X_i,X_j) \ge (\frac{1}{16}\gamma)^t$.
\end{lemma*}
\begin{proof}
We will prove the following claim:
\begin{claim*}
For any integer $s$ and $(2\lceil \frac{4}{\delta^2}\rceil + 1)$-local Boolean random variables $X_1,\ldots,X_n$, if $\E_{i \sim_s j}\pC(X_i,X_j) \ge \delta$, 
then $\E_{i \sim_{2s} j} \pC(X_i,X_j) \ge \frac{1}{4}\delta^2$.
\end{claim*}
We have assumed that $t = 2^\ell$, so we now apply this claim recursively $\ell$ times starting with $\delta = \gamma$, and we end up with a lower bound of
\[
\E_{i \sim_t j} \pC(X_i,X_j) \ge \frac{1}{4}\left(\frac{1}{4} \left(\cdots\frac{1}{4}\left( \frac{1}{4}\gamma^2\right)^2 \cdots\right)^2\right)^2 = \gamma^{2^\ell} \cdot \frac{1}{4} \cdot \frac{1}{4^2} \cdot \frac{1}{4^4} \cdots \frac{1}{4^{2^\ell}} \ge \gamma^{2^\ell} 4^{-2\cdot 2^\ell},
\]
which will give us the lemma.

Now, we prove the claim.
Choose a root vertex $r \sim \pi$, and let $i_1,\ldots,i_k$ be a set of $k = \lceil\frac{4}{\delta^2}\rceil$ vertices sampled independently, for each by taking an $s$-step random walk from $r$, $i_\ell \sim_s r$. 
Define the vectors $Z^{(1)},Z^{(-1)} \in \R^{2k+1}$ with the first $2k$ entries indexed by pairs $(j,a)$ for $j \in [k]$ and $a \in \{\pm 1\}$ and one singleton entry indexed by $r$, with $Z^{(b)}_{j,1} = \Ind[X_{i_j} = 1]$ and $Z^{(b)}_{j,-1} = \Ind[X_{i_j} = -1]$ for $j \in [k]$ and $Z^{(b)}_{r} = \Ind[X_r = b]$.
Let $C^{(1)},C^{(-1)} \in \R^{(2k+1) \times (2k+1)}$ be the covariance matrices
\[
C^{(b)} = \pE[Z^{(b)}(Z^{(b)})^\top] - \pE[Z^{(b)}]\pE[Z^{(b)}]^\top,
\]
for $b \in \{\pm 1\}$.
Since each $Z$ involves only $2k+1$ variables, the $C$ are covariance matrices of a true distribution and are thus positive semidefinite.
Further, let $u^{(1)},u^{(-1)} \in \R^{2k+1}$ be the vectors with 
\[
u^{(b)}_{S} =\begin{cases} \sgn\left(\Pr[X_{i_j} = a, X_r = b] - \Pr[X_{i_j} = a]\Pr[X_r = b]\right) & S = (j,a) \text{ for }j \in [k], a \in \{\pm 1\}\\
- \alpha & S = r 
\end{cases}
\]
for some $\alpha > 0$ to be chosen later.
Then we have that
\begin{align*}
0 &\le (u^{(1)})^\top C^{(1)} u^{(1)} + (u^{(-1)})^\top C^{(-1)} u^{(-1)}\\
& = 
\sum_{j,\ell \in [k]}\sum_{b \in \{\pm 1\}}\sum_{a,a' \in \{\pm 1\}}\left(\Pr[X_{i_j} = a, X_{i_{\ell}} = a'] - \Pr[X_{i_j} = a]\Pr[X_{i_{\ell}} = a']\right)u^{(b)}_{(j,a)}u^{(b)}_{\ell,a'}\\
&\qquad -2\alpha\sum_{j \in [k]} \sum_{a,b \in \{\pm 1\}} \left|\Pr[X_{i_j} = a, X_{r} = b] - \Pr[X_{i_j} = a]\Pr[X_{r} = b]\right|\\
&\qquad + \alpha^2 \sum_{b \in \{\pm 1\}} \left|\Pr[X_{r} = b] - \Pr[X_r = b]^2\right| \\
&\le 2\left(\sum_{j \neq \ell \in [k]} \pC(X_{i_j},X_{i_\ell}) \right) 
- 2\alpha \left(\sum_{j \in [k]} \pC(X_{i_j},X_r)\right)
+ \left(\alpha^2 \pV(X_r) + 2\sum_{j \in [k]} \pV(X_{i_j})\right),
\end{align*}
where the final inequality follows from the fact that the first $k$ entries of $u^{(b)}$ are signs.
Each pair $X_{i_j}, X_r$ is distributed identically to a pair from $i_j \sim_s r$, each pair $X_{i_j},X_{i_\ell}$ with $j \neq \ell$ is distributed identically to a pair from $i_j \sim_{2s}i_\ell$, and each $i_j,r$ are distributed according to $\pi$.
Thus, taking the expectation of the above inequality over a random choice of $r, i_1,\ldots,i_k$, we have
\[
0 \le
2k(k-1)\E_{i \sim_{2s} j} \pC(X_i,X_j) 
- 2\alpha k \E_{i \sim_s j} \pC(X_{i},X_j)
+ (\alpha^2 + 2k) \E_{i \sim \pi} \pV(X_i).
\]
Rearranging and simplifying, using $\pV(X_i) \leq 1$,
\[
\frac{\alpha \delta}{k} - \frac{\tfrac{1}{2}\alpha^2 + k}{k^2} \le \E_{i\sim_{2s} j} \pC(X_i,X_j)
\]
And choosing $\alpha = \delta k$ to maximize the left-hand side, as well as the assumption that $k \ge \frac{4}{\delta^2}$, we have our desired bound.
\end{proof}

Using Lemma~\ref{lem:loc-to-t} we can lower bound global correlation in terms of the local correlation and the trace.

\begin{lemma}\label{lem:loc-to-glob}
Let $\tau > 0$, and suppose $G$ is an $n$-vertex graph with no isolated vertices and with symmetric normalized adjacency matrix $N = D^{-\hf} A D^{-\hf}$. 
Let $X_1,\ldots,X_n$ be Boolean $T$-local random variables with $T \ge 2\left(\frac{16}{\delta}\right)^t+1$ for $t$ some power of two.
If the local correlation $\E_{(i,j)\in E} \pC(X_i,X_j) \ge \delta$, 
then the global squared correlation is lower bounded by
\[
\E_{i,j \sim \pi} \pC(X_i,X_j)^2 \ge \frac{\left(\frac{1}{16}\delta\right)^{2t}}{\Tr(N^{2t})}.
\]
\end{lemma}
\begin{proof}
By Lemma \ref{lem:loc-to-t}, the correlation of endpoints of $t$-length random walks is bounded by the local correlation,
\begin{align}
 \left(\frac{1}{16}\E_{i \sim j} \pC(X_i,X_j)\right)^{t}\label{eq:local-to-t}
\le \E_{i\sim_t j} \pC(X_i,X_j)
\end{align}
Further, letting $P = D^{-1} A$ be the transition matrix for the random walk on $G$ and letting $\pi$ be the stationary measure, we have that
\begin{align}
\E_{i\sim_t j} \pC(X_i,X_j)
&= \sum_{i,j} \pi_i\cdot (P^{t})_{i,j}\cdot \pC(X_i,X_j)\nonumber \\
&\le \sqrt{\left(\sum_{i,j} \frac{\pi_i}{\pi_j}(P^t)_{i,j}^2\right) \left(\sum_{i,j}\pi_i \pi_j \pC(X_i,X_j)^2\right)}\nonumber \\
&= \|D^{\hf} P^t D^{-\hf}\|_F \left(\E_{i,j \sim \pi}  \pC(X_i,X_j)^2 \right)^{\hf},\label{eq:frob-global}
\end{align}
where to obtain the first inequality we have applied Cauchy-Schwarz and the assumption that $G$ has no isolated vertices (which ensures $\pi_j \neq 0$).
Now, since the symmetric normalized adjacency matrix $N$ has the property that $N^t = D^{\hf} P^t D^{-\hf}$.
Therefore, $\|D^{\hf} P^t D^{-\hf}\|_F = \sqrt{\Tr(N^{2t})}$, and we can combine (\ref{eq:frob-global}) with (\ref{eq:local-to-t}) to deduce that
\[
\left(\frac{1}{16} \E_{i \sim j}\pC(X_i,X_j)\right)^{2t} \le \Tr(N^{2t}) \cdot \E_{i,j \sim \pi} \pC(X_i,X_j)^2,
\]
from which the lemma follows.
\end{proof}

\subsection{Trace bound for low threshold rank graphs}
\label{sec:trace-bd}
In the previous subsection, we obtained a bound on the global correlation in terms of the local correlation and the trace of a power of the normalized adjacency matrix.
Now, we will show a bound on the trace of powers of graphs in terms of the threshold rank.
Low threshold rank requires that there are few large-magnitude {\em positive} eigenvalues but does not provide explicit control over large-magnitude negative eigenvalues. This lemma provides a consequence of bounded threshold rank for the trace of powers of the normalized adjacency matrix.
\begin{claim}\label{claim:thresh-trace}
Let $N = D^{-\hf} A D^{-\hf}$ be the symmetric normalized adjacency matrix of a graph $G$.
If $\rank_{\tau}(G) \le k$, then for any integer $t \ge 1$, $\Tr(N^{2t}) \le 2(k + n \tau^{2t -1})$.
\end{claim}
\begin{proof}
Any symmetric matrix can be written as the sum of its projections on to the positive semidefinite and negative definite cones.
Let $N = N_+ + N_-$ where $N_+$ is the projection of $N$ to the positive semidefinite cone.
Then,
\[
\Tr(N_+^{2t + 1}) \le k + n\cdot \tau^{2t + 1},
\]
since there are at most $k$ eigenvalues $\ge \tau$ and $\|N\|\le 1$.
We also have that $\Tr(N^\ell) \ge 0$ for any integer $\ell$, since the entries of $N$ are nonnegative.
So, since $N = N_+ + N_-$, we have that
\[
0 \le \Tr(N^{2t + 1}) = \Tr(N_+^{2t+1}) + \Tr(N_-^{2t + 1}) \le (k + n \cdot \tau^{2t + 1}) + \Tr(N_-^{2t + 1}),
\]
which implies $\Tr(N_-^{2t + 1}) \ge - (k + n \cdot \tau^{2t + 1}) $.

Since $\|N\| \le 1$, $\Tr(N_-^{2t + 2}) \le |\Tr(N_-^{2t+1})|$.
The same is true for $N_+$, and from this we have
\[
\Tr(N^{2t+2}) \le 2(k + n \cdot \tau^{2t + 1})
\]
as desired.
\end{proof}

With Lemma \ref{lem:loc-to-glob}, this shows that the local and global correlation are related in low-threshold-rank graphs.

\subsection{Proof of main theorem}\label{sec:together}


With all these pieces we are ready to prove Theorem~\ref{th.main.maxcut}, restated more formally below.
(Theorem~\ref{th.low-threshold-rank} follows from a subset of the proof below, omitting the graph partitioning argument.)

\begin{theorem}[Sherali-Adams for Max-Cut]\label{thm:main-mc}
For every $\alpha > 0$ there is an $\e_\alpha = \exp(-O(1/\alpha^3))$ such that if $G$ is an $n$ node graph with $n$ sufficiently large and $\SA_{n^{\alpha}}(G) \geq 1 -\e_\alpha$, then there is a cut which cuts at least a $(\tfrac 1 2 + \eps_\alpha)$-fraction of the edges of $G$, and there is a polynomial time rounding algorithm for the degree-$n^{\alpha}$ Sherali-Adams LP that produces such a cut.
Consequently, the degree-$n^{\alpha}$ Sherali-Adams LP value provides a $(\frac{1}{2} +\eps_\alpha)$-approximation to Max-Cut.
\end{theorem}
\begin{proof}
Let $\e_\alpha$ be a small number to be set later.
Let $X_1',\ldots,X_n'$ be $T$-local random variables for $T \ge n^\alpha$, and let $\delta = \tfrac{1}{64}$ and $\tau = \left(\frac{\delta}{16}\right)^{2/\alpha}$.
Assume that the objective value is $\SA(G) = \E_{i \sim j} \E_{X_i',X_j'} \Ind(X_i' \neq X_j') \geq 1 - \e_\alpha$.

If $\rank_\tau(G) \le n^{\alpha/2}$, then we may apply global correlation rounding.
We sample $i_1,\ldots,i_k \sim \pi$, where $\pi$ is the stationary measure on $G$, for some $k \leq n^{\alpha}/2$, and then sample values for $X_{i_1}',\ldots,X_{i_k}'$ according to their local distribution.
Conditioning on those values, we obtain $(T - k)$-local Boolean random variables $X_1,\ldots,X_{n}$, we may assume that both $\E_{i,j \sim \pi} \pC (X_i,X_j)^2 \leq 12 n^{-\alpha}$ and that $\E_{i \sim j} \E_{X_i,X_j} \Ind(X_i \neq X_j) \geq 1 - 10\e_\alpha$, as guaranteed by Theorem \ref{thm:gcr-conditioning} together with Markov's inequality.
We now apply Lemma \ref{lem:loc-to-glob} with $t$ the largest power of two such that
\[
\frac{(1-\tfrac{1}{2}\alpha)\log n}{2\log \frac{1}{\tau}} \le t \le \frac{\alpha\log n}{2\log \frac{16}{\delta}}.
\]
It may be checked that our choice of $\delta$ and $\tau$ ensures that such a $t$ exists.
Further, this guarantees that $2\left(\frac{16}{\delta}\right)^t \le n^\alpha$ (so that the locality of our local random variables is high enough to apply \ref{lem:loc-to-glob}) and also $n \tau^{2t-1} \le n^{\alpha/2}$.
From Lemma \ref{lem:loc-to-glob} we have
\[
\left(\frac{1}{16}\E_{i \sim j}\pC(X_i,X_j)\right)^{2t} 
\le \Tr(N^{2t}) \cdot \E_{i,j} \pC(X_i,X_j)^2
\le \Tr(N^{2t}) \cdot 12 n^{-\alpha}. 
\]
Since $\rank_\tau(G) \le n^{\alpha/2}$ and since by our choice of $t$ we have $n\tau^{2t -1} \le n^{\alpha/2}$ we can apply Claim \ref{claim:thresh-trace} to the above to obtain
\[
\left(\frac{1}{16}\E_{i \sim j}\pC(X_i,X_j)\right)^{2t}
\le
2(\rank_\tau(G) + n\tau^{2t-1})\cdot 12 n^{-\alpha}
\le 48 n^{-\alpha/2}.
\]
This implies that the local correlation is at most
\[
\E_{i\sim j} \pC(X_i,X_j) \le 16\cdot \left(48 n^{-\alpha/2}\right)^{1/t} \le 2\delta,
\] 
for $n$ sufficiently large.
Applying Lemma \ref{lem:gcr-rounding} we can round to obtain a solution of value at least $\SA(G) - 2\delta$.

Otherwise, if $\rank_\tau(G) > n^{\alpha/2}$, we apply Theorem \ref{thm:partition} to obtain a partition of $G$ into pieces $G_1,\ldots,G_m$ of threshold rank at most $n^{\alpha/2}$, such that the partition has expansion at most $1 - \exp(-O(\frac{1}{\alpha^2} \log \frac{1}{\tau})) = 1 - \exp(-C\frac{1}{\alpha^3})) \le 1 - \sqrt{\eps_{\alpha}}$ (where $C$ is a universal constant that comes from the partitioning theorem, and we have chosen the constant in the theorem statement to make this equality hold).
Since each piece has threshold rank at most $n^{\alpha/2}$, we may apply global correlation rounding to each piece as above to obtain an assignment $x^{(i)}$ on the variables of $G_i$ which obtains value $\ge \SA(G_i) - 2\delta$ within $G_i$.
Furthermore, because the objective value $\SA(G)$ remains at least $1 - 2\eps_{\alpha}$ even after conditioning, and since $\cup_{i \in [m]} |E[G_i]|$ accounts for a $\ge \sqrt{\eps_\alpha}$ fraction of the total edges in the graph, on average $\SA(G_i) \ge 1 - 2\sqrt{\eps_\alpha}$ and so we can round the Sherali--Adams solution within each piece to obtain a solution $x^{(i)}$ so that on average, the value of $x^{(i)}$ within $G_i$ is at least $1 - 2\sqrt{\eps_\alpha} - 2\delta$. 
That is, on the $\Omega(\eps_{\alpha})$ fraction of edges {\em not} cut by the partition, we have produced a max-cut of value close to the original LP value, which itself is close to $1$.
Now, we will use randomization to make sure that the edges cut by the partition get value at least $\frac{1}{2}$.

Choosing random signs $s_1,\ldots,s_m$ for each piece of the partition and choosing $s_1 x^{(1)},\ldots, s_m x^{(m)}$ as our global solution will give a solution of expected value $\ge \frac{1}{2}\cdot (1-\eps_\alpha) + (1-2\sqrt{\eps_{\alpha}} - 2\delta)\cdot \eps_{\alpha} = \frac{1}{2} + (\frac{1}{2}-2\delta)\eps_{\alpha} - 2\eps_{\alpha}^{3/2}$, since each edge crossing a partition is cut with probability $\frac{1}{2}$ and the value within the union of the parts is at least $1-2\sqrt{\eps_\alpha} - 2\delta$.
We can think of the partition as defining a new max-cut instance in which the $s^{(i)}$ are the max cut variables, so this can be derandomized in polynomial time by applying standard arguments (e.g. using the greedy algorithm).
By our choice of $\eps_{\alpha}$, we have also that $2\eps_{\alpha}^{3/2} < \frac{1}{8}\eps_\alpha$, and we chose $\delta \le \frac{1}{64}$, from which we obtain the objective value promised in the theorem statement (after rescaling $\eps_{\alpha}$).
This completes the proof.
\end{proof}

\section{Unique Games}
\label{sec:ug}
For the case of unique games, we need a slight modification of the proof of the Max-Cut result.
Rather than using the $\ell_1$ distance between joint and marginal distributions of variables as our notion of correlation, we will use the notion of permutation-correlation defined in Definition~\ref{def:pcor}.
We will make use of the following lemma, which we reproduce here for completeness:
\begin{lemma}\label{lem:gcr-ug}
Suppose that $(G,\Pi)$ is an instance of unique games on $2$-local variables $X_1,\ldots,X_n$ with constraints $\Pi$ given by $\sigma_{ij}:[q]\to[q]$ on each edge $(i,j) \in G$.
Suppose that $\E_{i\sim j} \pC_\pi(X_i,X_j) \le \delta$.
Then a random assignment $Y_1,\ldots,Y_n$ in which each $Y_i$ is sampled independently according to the marginals of $\{X_i\}$ has expected value at least
\[
\E_Y \val_G(Y) \ge \pE[\val_G(X)] - \delta.
\]
\end{lemma}
\begin{proof}
By definition of unique games we have that
\begin{align*}
\E_Y \val_G(Y)
&= \E_{i\sim j} \sum_{a \in [q]} \E_Y[\Ind(Y_i = a, Y_j = \sigma_{ij}(a))]\\
&= \E_{i \sim j} \sum_{a \in [q]} \pE[\Ind(Y_i = a)]\pE[\Ind(Y_j=\sigma_{ij}(a))]\\
&\ge \pE[\val_G(X)] - \E_{i \sim j} \sum_{a \in [q]} \left|\pE[\Ind(Y_i = a, Y_j = \sigma_{ij}(a))] - \pE[\Ind(Y_i = a)]\pE[\Ind(Y_j=\sigma_{ij}(a)] \right|\\
&\ge \pE[\val_G(X)] - \E_{i \sim j}\max_{\pi \in \cS_q} \sum_{a \in [q]} \left|\pE[\Ind(Y_i = a, Y_j = \pi(a))] - \pE[\Ind(Y_i = a)]\pE[\Ind(Y_j=\pi (a)]\right|\\
&= \pE[\val_G(X)] - \E_{i\sim j}\pC_\pi(X_i,X_j).
\end{align*}
The conclusion follows immediately.
\end{proof}

We can relate the local permutation-correlation to the global permutation-correlation via a slight modification of the proof of Lemma \ref{lem:loc-to-glob}.

\begin{lemma}\label{lem:loc-t-ug}
Let $G$ be a graph on $n$ vertices, let $\ell = 2^t$ be a power of two, and let $X_1,\ldots,X_n$ be $T$-local $q$-ary random variables variables with $T \ge \lceil(\frac{4}{\gamma})^\ell\rceil + 1$.
Then $\E_{i\sim j} \pC_\pi(X_i,X_j) \ge \gamma$ implies that $\E_{i \sim_\ell j}\pC_\pi(X_i,X_j) \ge (\frac{1}{4}\gamma)^\ell$.
\end{lemma}
\begin{proof}
As for Lemma \ref{lem:loc-to-t}, the proof boils down to the following claim:
\begin{claim}
If $X_1,\ldots,X_n$ are $(\lceil \frac{2}{\delta^2}\rceil + 1)$-local random variables then $\E_{i \sim_t j} \pC_\pi(X_i,X_j) \ge \delta$ implies $\E_{i \sim_{2t} j} \pC_\pi(X_i,X_j) \ge \frac{1}{2}\delta^2$.
\end{claim}
From this claim we conclude the lemma in a manner identical to the proof of Lemma \ref{lem:loc-to-glob}.

To prove the claim, consider the following procedure for sampling $k+1$ variables: first the root vertex is sampled $r \sim \pi$, then $i_1,\ldots,i_k$ are sampled by taking a random walk in $G$ of length $\ell$ from $r$, so that $i_{\ell} \sim_t r$ for each $\ell \in [k]$.
We now define $q$ vectors $Z^{(1)},\ldots,Z^{(q)} \in \R^{k+1}$,
and $k$ permutations $p_j:[q]\to[q]$, one for each $j \in [k]$, so that 
\[
p_j = \argmax_\pi \sum_{a \in [q]} |\Ind[X_r = a, X_{i_j} = \pi(a)] - \Ind[X_r = a]\Ind[X_{i_j} = \pi(a)]|,
\]
For each $j \in [k]$,
we let $Z^{(a)}_{j} = \Ind[X_{i_j} = p_j(a)]$, and $Z^{(a)}_{k+1} = \Ind[X_r = a]$.
If the variables are $(k+1)$-local, we have that
\[
M^{(a)} = \pE[Z^{(a)}(Z^{(a)})^\top] - \pE[Z^{(a)}]\pE[Z^{(a)}]^\top \succeq 0.
\]

Now, we define the $q$ vectors $u^{(1)},\ldots,u^{(q)} \in \R^{k+1}$ by taking
\[
u^{(a)}_{j} = \begin{cases}
\sgn\left(\pE[\Ind[X_{r} = a, X_{i_j} = p_j(a)]] - \pE[X_r = a]\pE[X_{i_j} = p_j(a)]\right) &  j \in [k]\\
- \alpha & j = k+1.
\end{cases}
\]

For shorthand, let $\pC(X_{r,a},X_{i_j,b}) = |\pE[\Ind[X_r = a,X_{i_j} = b] - \pE[\Ind[X_r = a]]\pE[X_{i_j} = b]|$.
By the positive-semidefiniteness of the $M^{(a)}$, we have that
\begin{align*}
 0 
&\le \sum_{a = 1}^q (u^{(a)})^\top M^{(a)} u^{(a)} \\
&\le \sum_{a = 1}^q \left( \alpha^2 \pV(X_{r,a}) - 2\alpha\sum_{j \in [k]}\pC(X_{r,a},X_{i_j,p_j(a)}) + \sum_{j,\ell \in [k]} \pC(X_{i_j,p_j(a)},X_{i_\ell,p_\ell(a)})\right)\\
&= \alpha^2 \pV(X_{r}) - 2\alpha\sum_{j \in [k]} \pC_\pi(X_r,X_j) + \sum_{j,\ell \in [k]} \sum_{a = 1}^q \left( \pC(X_{i_j,p_j(a)},X_{i_\ell,p_\ell(a)})\right)\\
&\le \alpha^2 \pV(X_{r}) + \sum_{j \in [k]} \pV(X_{i_j}) - 2\alpha\sum_{j \in [k]} \pC_\pi(X_r,X_j) + \sum_{j\neq \ell \in [k]} \pC_\pi(X_{i_j},X_{i_\ell}),
\end{align*}
where in the third line we have used that $p_j$ were selected to be the maximizers of the permutation-correlation, and in the final line we have used that the permutation-correlation is defined with respect to the maximum permutation (which achieves a value at least as high as $p_j \circ p_\ell^{-1}$).
We have also used the shorthand $\pV(X_i) = \sum_{a \in [q]} \pV(X_{i_a})$.
Taking the expectation over $r,i_1,\ldots,i_k$,
\begin{align*}
0 &\le k(k-1) \E_{i \sim_{2t} j} \pC(X_i,X_j).
- 2\alpha k \E_{i \sim_t j} \pC_\pi(X_i,X_j) + (\alpha^2 + k)\E_{i \sim \pi} \pV(X_i)
\end{align*}
Now, using that $\pV(X_i) \le 1$ as well as our lower bound $\delta \le \E_{i\sim_t j} \pC_\pi(X_i,X_j)$, we re-arrange and simplify
\begin{align*}
\frac{2\alpha \delta}{k^2} - \frac{\alpha^2 + k}{k^2} 
&\le \E_{i \sim_{2t} j} \pC_\pi(X_i,X_j).
\end{align*}
Taking $\alpha = \delta k$ and using the assumption that $k \ge 2 \delta^{-2}$ gives us our desired bound.
\end{proof}

We can then obtain the following lower bound as a corollary:
\begin{lemma}\label{lem:loc-glob-ug}
Let $\tau > 0$, and suppose $G$ is a simple $n$-vertex graph with no isolated vertices with symmetric normalized adjacency matrix $N = D^{-\hf} A D^{-\hf}$. 
Let $X_1,\ldots,X_n$ be $(\left(\frac{4}{\delta}\right)^t+1)$-local variables supported on $[q]$ (for $t$ some power of $2$). 
If the local permutation correlation $\E_{(i,j)\in E} \pC_\pi (X_i,X_j) \ge \delta$, 
Then the global squared permutation correlation is lower bounded by
\[
\E_{i,j \sim \pi} \pC_\pi(X_i,X_j)^2 \ge \frac{\left(\frac{1}{4}\delta\right)^{2t}}{\Tr(N^{2t})}.
\]
\end{lemma}
\begin{proof}
The proof is identical to that of Lemma \ref{lem:loc-to-glob}, except that it uses Lemma \ref{lem:loc-t-ug} in place of Lemma \ref{lem:loc-to-t}.
\end{proof}

We can now prove a theorem for unique games, with an argument almost identical to that of our theorem \ref{thm:main-mc} for max cut.

\begin{theorem}[Sherali-Adams for Unique Games]
For every $\alpha > 0$ there exists $\e_\alpha = \exp(-O(1/\alpha^3))$ such that for every $q \in \N$,
and big-enough $n$, if $(G,\Pi)$ is an $n$-variable instance of Unique Games with alphabet size $q$ and $SA_{2n^\alpha\log q }(G,\Pi) \geq 1-\e_\alpha$ then there exists $x_1,\ldots,x_n \in [q]^n$ such that $\E_{i \sim j} \Pi(x_i,x_j) \geq \e_\alpha/8$.
Furthermore, there is a polynomial-time rounding algorithm which takes a Sherali-Adams pseudodistribution of degree $2\log q\cdot n^\alpha$ and objective value at least $1-\e_\alpha$ and finds such a solution $x$.
\end{theorem}
\begin{proof}
Let $X_1,\ldots,X_n$ be $T$-local random variables supported on $[q]$ with $T \ge (2\log q) n^\alpha$, and let $\delta = \tfrac{1}{16}$ and $\tau = \left(\frac{\delta}{4}\right)^{1/\alpha}$.

If $\rank_\tau(G) \le n^{\alpha/2}$, then we may apply global correlation rounding.
We condition on $n^{\alpha}\log q$ variables chosen uniformly at random from $\pi$, then in expectation the global mutual information drops so that $\E_{i,j \sim \pi} I(X_j;X_i) \le 3 n^{-\alpha}$ as guaranteed by Theorem \ref{thm:gcr-conditioning}, and the objective value remains at least $1 - 2\eps$.
From Pinsker's inequality, we then have that $\E_{i,j \sim \pi} \pC_\pi(X_i,X_j)^2\le 6 n^{-\alpha}$.
We now apply Lemma \ref{lem:loc-glob-ug} with $t$ the largest power of two such that
\[
\frac{(1-\tfrac{1}{2}\alpha)\log n}{2\log \frac{1}{\tau}} \le t \le \frac{\alpha\log n}{\log \frac{4}{\delta}}.
\]
It may be checked that our choice of $\delta$ and $\tau$ ensures that such a $t$ exists.
Further, this guarantees that $\left(\frac{4}{\delta}\right)^t + 1 \le \log q \cdot n^\alpha$ (so that the Sherali--Adams pseudodistribution has sufficiently high degree to apply Lemma \ref{lem:loc-glob-ug}) and also $n \tau^{2t-1} \le n^{\alpha/2}$.
Combining Lemma \ref{lem:loc-glob-ug}, our bound $\rank_\tau(G) \le n^{\alpha/2}$, and the fact that by our choice of $t$ we have $n\tau^{2t -1} \le n^{\alpha/2}$ we can apply Claim \ref{claim:thresh-trace} to the above to obtain
\[
\left(\frac{1}{4}\E_{i \sim j}\pC_\pi(X_i,X_j)\right)^{2t}
\le
2(\rank_\tau(G) + n\tau^{2t-1})\cdot 6 n^{-\alpha}
\le 24 n^{-\alpha/2}.
\]
This implies that the local permutation correlation is at most 
\[
\E_{i\sim j} \pC_\pi(X_i,X_j) \le 4\cdot \left(24 n^{-\alpha/2}\right)^{1/t} \le 2\sqrt{\delta}.
\] 
Applying Lemma \ref{lem:gcr-ug} we can round to obtain a solution of value at least $\SA(G,\Pi) - 2\sqrt{\delta}$.

Otherwise, if $\rank_\tau(G) > n^{\alpha/2}$, we apply Theorem \ref{thm:partition} to obtain a partition of $G$ into pieces $G_1,\ldots,G_m$ of threshold rank at most $n^{\alpha/2}$, such that the partition has expansion at most $1 - \exp(-O(\frac{1}{\alpha^2} \log \frac{1}{\tau})) \le 1 - \exp(-C\frac{1}{\alpha^3})) = 1 - \sqrt{\eps_{\alpha}}$ (where $C$ is the universal constant mentioned in the theorem statement, which comes from the partitioning theorem).
Since each piece has threshold rank at most $n^{\alpha/2}$, we may apply global correlation rounding to each piece as above to obtain an assignment $x^{(i)}$ on the variables of $G_i$ which obtains value $\ge \SA(G_i,\Pi) - 2\sqrt{\delta}$ within $G_i$.
Furthermore, because the objective value $\SA(G,\Pi)$ remains at least $1 - 2\eps$ even after conditioning for $\eps \le \eps_{\alpha}$, and since $\cup_{i \in [m]} |E[G_i]|$ accounts for a $\ge \sqrt{\eps_\alpha}$ fraction of the total edges in the graph, on average $\SA(G_i,\Pi) \ge 1 - 2\sqrt{\eps_\alpha}$ and so we can round the Sherali--Adams solution within each piece to obtain a solution $x^{(i)}$ so that on average, the value of $x^{(i)}$ within $G_i$ is at least $1 - 2\sqrt{\eps_\alpha} - 2\sqrt{\delta}$. 
Now, we take the global solution $x^{(1)},\ldots, x^{(m)}$ which will have value at least $\ge (1-2\sqrt{\eps_{\alpha}} - 2\sqrt{\delta})\cdot \eps_{\alpha}$.
By our choice of $C$, we have also that $2\eps_{\alpha}^{3/2} < \frac{1}{8}\eps_\alpha$, and we chose $\delta \le \frac{1}{16}$, from which we obtain the objective value promised in the theorem statement.
This completes the proof.
\end{proof}

\subsection{Generalization to Permutation-Symmetric CSPs}

In this section we describe how to modify the arguments in preceding sections to prove the following theorem.

\begin{theorem}
  Suppose $P_1,\ldots,P_k \, : \, [q] \times [q] \rightarrow \{0,1\}$ is a family of predicates such that there is a set of permutations $S \subseteq S_q$ such that
  \begin{itemize}
  \item for every $a \in [q]$, if $\pi \sim S$ is chosen uniformly at random then $\pi(a)$ is uniform in $[q]$, and
  \item for every $\pi \in S$ and every $P_i$ and $a,b$ such that $P_i(a,b) = 1$, also $P_i(\pi(a),\pi(b)) = 1$.
  \end{itemize}
  For every $\alpha > 0$ there is $\epsilon(\alpha,q) > 0$ such that Sherali-Adams of degree $n^{\alpha}$ provides a $c + \epsilon$ approximation to Max-$P$, (i.e. the special case of 2-CSP-$q$ where all the predicates $\Pi$ are chosen from $P_1,\ldots,P_m$) where $c = \min_i \E_{a,b \sim [q]} P_i(a,b)$ is the approximation ratio achieved by the random assignment.
\end{theorem}

\begin{remark}
Notice that Max-2-Lin and Max-$k$-Cut both fit into this framework; for $2$-Lin solutions are invariant under linear shifts, and for $k$-Cut solutions are invariant under relabeling the partition.
\end{remark}

The proof proceeds as in the case of Max-Cut, with the twist that we use the following local-to-global correlation bound: if $G$ is a graph on $[n]$, $\ell = 2^t$ is a power of two, and $X_1,\ldots,X_n$ are $T$-local $q$-ary random variables with $T \geq q^{O(\ell)} /\gamma^{O(\ell)}$ and $\E_{i \sim j} \Cov(X_i,X_j) \geq \gamma$ then $\E_{i \sim_\ell j} \Cov(X_i,X_j) \geq O(\gamma/q)^\ell$.
To prove this fact it suffices to apply Lemma~\ref{lem:loc-t-ug} together with the observation that by a simple averaging argument, for any $[q]$-valued random variables $X,Y$ it holds that $\Cov_\pi(X,Y) \geq \Cov(X,Y)/q$.

The proof then proceeds as before.
We may assume given an instance $(G,\Pi)$ of Max-CSP where all predicates $\Pi$ are chosen from among $P_1,\ldots,P_k$ and a collection of $n^{\alpha}$-local random variables $X_1,\ldots,X_n$ such that $\E_{i \sim j} \E_{X_i,X_j} \Pi_{ij}(X_i,X_j) \geq 1 - \e'$ for as small an $\e'(\e,q)$ as we like.
We partition $G$ into subgraphs of threshold-rank $n^{\alpha/2}$ with threshold $\tau = \tau(\alpha,q)$ a small-enough quantity to be chosen later.
This means that at an least $\e''(\alpha,q)$-fraction of edges of $G$ are preserved in the partition.
On each component $H$ we can perform conditioning to obtain conditioned $n^{\alpha/2}$-local random variables $Y_1,\ldots,Y_n$ with global correlation $\E_{i,j \sim \pi_H} \Cov(Y_i,Y_j) \leq \log q / n^{-\alpha/2}$.
Via the local-to-global correlation result above together with small-enough choice of $\tau$, this implies that coordinate-wise rounding of the $Y$'s produces an assignment which satisfies $1-\e'''$ of the edges in $G_1,\ldots,G_m$.
Let $Z_1,\ldots,Z_n$ be this coordinate-wise rounding.
Let $Z_1',\ldots,Z_n'$ be the result of choosing a permutation $\pi \in S$ to apply to each of the components $G_i$.
By hypothesis $Z_1',\ldots,Z_n'$ now satisfy a $1-\e'''$ fraction of edges in $G_1,\ldots,G_m$ and at least a $\min_i \E_{a,b \sim [q]} P_i(a,b)$ fraction of edges between $G_i$'s.

\section{Acknowledgments}
TS would like to thank Ryan O'Donnell for the collaboration on \cite{ODonnellS19} and the questions and follow-up conversations which were crucial for the present work.
SBH thanks Prasad Raghavendra for helpful conversations as this work was being prepared.
SBH was supported by a UC Berkeley Miller Fellowship.
TS was supported by NSF grants CCF 1565264 and CNS 1618026.
LT  was supported by the NSF under grant CCF 1815434 and his work on this project has received funding from the European Research Council (ERC) under the European Union’s Horizon 2020 research and innovation programme (grant agreement No. 834861).

\bibliographystyle{amsalpha}
\bibliography{refs}

\newcommand{\etalchar}[1]{$^{#1}$}
\providecommand{\bysame}{\leavevmode\hbox to3em{\hrulefill}\thinspace}
\providecommand{\MR}{\relax\ifhmode\unskip\space\fi MR }
\providecommand{\MRhref}[2]{%
  \href{http://www.ams.org/mathscinet-getitem?mr=#1}{#2}
}
\providecommand{\href}[2]{#2}
\begin{thebibliography}{AKK{\etalchar{+}}08}

\bibitem[ABS15]{ABS}
Sanjeev Arora, Boaz Barak, and David Steurer, \emph{Subexponential algorithms
  for unique games and related problems}, J. {ACM} \textbf{62} (2015), no.~5,
  42:1--42:25.

\bibitem[AKK{\etalchar{+}}08]{AroraKKSTV08}
Sanjeev Arora, Subhash Khot, Alexandra Kolla, David Steurer, Madhur Tulsiani,
  and Nisheeth~K. Vishnoi, \emph{Unique games on expanding constraint graphs
  are easy: extended abstract}, Proceedings of the 40th Annual {ACM} Symposium
  on Theory of Computing, Victoria, British Columbia, Canada, May 17-20, 2008
  (Cynthia Dwork, ed.), {ACM}, 2008, pp.~21--28.

\bibitem[BFPS12]{BFPS}
Gabor Braun, Samuel Fiorini, Sebastian Pokutta, and David Steurer,
  \emph{Approximation limits of linear programs (beyond hierarchies)},
  Proceedings of the 2012 IEEE 53rd Annual Symposium on Foundations of Computer
  Science (Washington, DC, USA), FOCS '12, IEEE Computer Society, 2012,
  pp.~480--489.

\bibitem[BRS11]{BRS11}
Boaz Barak, Prasad Raghavendra, and David Steurer, \emph{Rounding semidefinite
  programming hierarchies via global correlation}, {IEEE} 52nd Annual Symposium
  on Foundations of Computer Science, {FOCS} 2011, Palm Springs, CA, USA,
  October 22-25, 2011 (Rafail Ostrovsky, ed.), {IEEE} Computer Society, 2011,
  pp.~472--481.

\bibitem[CMM09]{CMM09}
Moses Charikar, Konstantin Makarychev, and Yury Makarychev, \emph{Integrality
  gaps for {S}herali-{A}dams relaxations}, Proceedings of the 41st ACM
  Symposium on Theory of Computing, 2009, pp.~283--292.

\bibitem[CW04]{CW04}
Moses Charikar and Anthony Wirth, \emph{Maximizing quadratic programs:
  Extending grothendieck's inequality}, 45th Annual IEEE Symposium on
  Foundations of Computer Science, IEEE, 2004, pp.~54--60.

\bibitem[FKP19]{FKP19}
Noah Fleming, Pravesh Kothari, and Toniann Pitassi, \emph{Semialgebraic proofs
  and efficient algorithm design}.

\bibitem[GS11]{GS11}
Venkatesan Guruswami and Ali~Kemal Sinop, \emph{Lasserre hierarchy, higher
  eigenvalues, and approximation schemes for graph partitioning and quadratic
  integer programming with psd objectives}, 2011 IEEE 52nd Annual Symposium on
  Foundations of Computer Science, IEEE, 2011, pp.~482--491.

\bibitem[GW95]{GW}
Michel~X. Goemans and David~P. Williamson, \emph{Improved approximation
  algorithms for maximum cut and satisfiability problems using semidefinite
  programming}, J. ACM \textbf{42} (1995), no.~6, 1115--1145.

\bibitem[H{\aa}s08]{H08}
Johan H{\aa}stad, \emph{Every 2-csp allows nontrivial approximation},
  Computational Complexity \textbf{17} (2008), no.~4, 549--566.

\bibitem[Kho02]{K02}
Subhash Khot, \emph{On the power of unique 2-prover 1-round games}, Proceedings
  of the thiry-fourth annual ACM symposium on Theory of computing, ACM, 2002,
  pp.~767--775.

\bibitem[KMR17]{KMR}
Pravesh~K. Kothari, Raghu Meka, and Prasad Raghavendra, \emph{Approximating
  rectangles by juntas and weakly-exponential lower bounds for {LP} relaxations
  of {CSP}s}, Proceedings of the 49th Annual {ACM} {SIGACT} Symposium on Theory
  of Computing, {STOC} 2017, 2017, pp.~590--603.

\bibitem[LR99]{LR99}
Tom Leighton and Satish Rao, \emph{Multicommodity max-flow min-cut theorems and
  their use in designing approximation algorithms}, Journal of the ACM (JACM)
  \textbf{46} (1999), no.~6, 787--832.

\bibitem[LS91]{LS91}
L{\'a}szl{\'o} Lov{\'a}sz and Alexander Schrijver, \emph{Cones of matrices and
  set-functions and 0--1 optimization}, SIAM journal on optimization \textbf{1}
  (1991), no.~2, 166--190.

\bibitem[OS19]{ODonnellS19}
Ryan O'Donnell and Tselil Schramm, \emph{Sherali - adams strikes back}, 34th
  Computational Complexity Conference, {CCC} 2019, July 18-20, 2019, New
  Brunswick, NJ, {USA.} (Amir Shpilka, ed.), LIPIcs, vol. 137, Schloss Dagstuhl
  - Leibniz-Zentrum fuer Informatik, 2019, pp.~8:1--8:30.

\bibitem[RS09]{RS09}
Prasad Raghavendra and David Steurer, \emph{Integrality gaps for strong sdp
  relaxations of unique games}, Proceedings of the 2009 50th Annual IEEE
  Symposium on Foundations of Computer Science (Washington, DC, USA), FOCS '09,
  IEEE Computer Society, 2009, pp.~575--585.

\bibitem[RT12]{RagT12}
Prasad Raghavendra and Ning Tan, \emph{Approximating csps with global
  cardinality constraints using {SDP} hierarchies}, Proceedings of the
  Twenty-Third Annual {ACM-SIAM} Symposium on Discrete Algorithms, {SODA} 2012,
  Kyoto, Japan, January 17-19, 2012 (Yuval Rabani, ed.), {SIAM}, 2012,
  pp.~373--387.

\bibitem[SA90]{SA90}
Hanif~D Sherali and Warren~P Adams, \emph{A hierarchy of relaxations between
  the continuous and convex hull representations for zero-one programming
  problems}, SIAM Journal on Discrete Mathematics \textbf{3} (1990), no.~3,
  411--430.

\bibitem[Ste10a]{Steurer-thesis}
David Steurer, \emph{On the complexity of unique games and graph expansion},
  Ph.D. thesis, Princeton University, 2010.

\bibitem[Ste10b]{Steurer10}
David Steurer, \emph{Subexponential algorithms for d-to-1 two-prover games and
  for certifying almost perfect expansion}, Available at the author’s website
  \textbf{1} (2010), 2--1.

\bibitem[STT07]{STT}
Grant Schoenebeck, Luca Trevisan, and Madhur Tulsiani, \emph{Tight integrality
  gaps for {L}ovasz-{S}chrijver {LP} relaxations of vertex cover and max cut},
  Proceedings of the 39th ACM Symposium on Theory of Computing, 2007,
  pp.~302--310.

\bibitem[Tre09]{T}
Luca Trevisan, \emph{Max {C}ut and the smallest eigenvalue}, Proceedings of the
  41st ACM Symposium on Theory of Computing, 2009, pp.~263--272.

\end{thebibliography}

\appendix
\section{Sherali-Adams for Max-QP}
\label{sec:qp}
\newcommand{\mper}{\, .}

In this section we prove Observation~\ref{obs.main.qp}, which follows from the following two facts.

\begin{fact}
  Let $A \in \R^{n \times n}$ be a matrix with $A_{ii} = 0$ for $i \in [n]$.
  Let $X_1,\ldots,X_n$ be Boolean $k$-local random variables.
  Then there exists $x \in \{ \pm 1 \}^n$ such that $x^\top A x \geq \Omega(\tfrac k n) \sum_{i,j \leq n} A_{ij} \E X_i X_j$.
\end{fact}
\begin{proof}
  Without loss of generality we can assume $\sum_{i,j \leq n} A_{ij} \E X_i X_j \geq 0$, otherwise any $x$ such that $x^\top A x \geq 0$ will do. (Such $x$ must exist since $\E_{x \sim \{\pm 1\}^n} x^\top Ax = 0$.)
  Let $S_1,\ldots,S_k$ be a uniformly random partition of $[n]$ into sets of size $n/k$.
  Then
  \[
  \E_{S_1,\ldots,S_k} \sum_{\ell \leq n/k} \sum_{i,j \in S_\ell} A_{ij} \E X_i X_j = \sum_{ij \leq n} \Pr(i,j \text{ in same } S_\ell) A_{ij} \E X_i X_j \geq \Omega(k/n) \sum_{ij \leq n} A_{ij} \E X_i X_j\mper
  \]
  Fix some $S_1,\ldots,S_k$ such that
  \[
  \sum_{\ell \leq n/k} \sum_{ij \in S_\ell} A_{ij} \E X_i X_j \geq \Omega(k/n) \sum_{ij \leq n} A_{ij} \E X_i X_j\mper
  \]
  Since $|S_\ell| = k$ and $X_1,\ldots,X_n$ are $k$-local, for each $\ell \leq n/k$ there exists $z \in \{\pm 1\}^{S_\ell}$ such that $\sum_{ij \in S_\ell} A_{ij} z_i z_j \geq \sum_{ij \in S_\ell} A_{ij} \E X_i X_j$.

  We define a family of $\{ \pm 1\}$-valued random variables $Z_1,\ldots,Z_n$ as follows.
  First for each $\ell \leq n/k$, sample a random $s_\ell \in \{ \pm 1\}$.
  Then for each $i \in S_\ell$, let $Z_i = s_\ell \cdot z_i$, where $z \in \{ \pm 1\}^{S_\ell}$ is as above.

  Note that
  \[
  \E_Z \sum_{ij} A_{ij} Z_i Z_j = \sum_{\ell \leq n/k} \sum_{ij \in S_\ell} A_{ij} z_i z_j \geq \Omega(k/n)\sum_{i,j \leq n} A_{ij} \E X_i X_j\mper
  \]
  The fact follows.
\end{proof}

\begin{fact}
  For $n \in \N$ and $k \leq n$ there exist Boolean $k$-local random variables $X_1,\ldots,X_n$ such that $\sum_{i \neq j \leq n} - \E X_i X_j \geq \Omega(n^2/k)$.
  Moreover, $\max_{x \in \{ \pm 1\}^n} \sum_{i \neq j \leq n} - x_i x_j \leq O(n)$.
\end{fact}
\begin{proof}
  We start with the second claim.
  If $1$ is the $n$-dimensional all-$1$'s vector, then we have $\sum_{i \neq j} -x_i x_j = -\iprod{x,1}^2 + \sum_i x_i^2 \leq n$.

  Now we construct the $k$-local random variables $X_1,\ldots,X_n$.
  We start with the following claim.
  Suppose $\mu$ is a distribution on $\{ \pm 1\}^k$ with the following symmetry property: for all $t \leq k$ and $S,S' \subseteq [k]$ with $|S| = |S|' = t$ we have $\E_{x \sim \mu} x^S = \E_{x \sim \mu} x^{S'} = c_t$ for some numbers $c_0,\ldots,c_k \in \R$.
  Then there are $k$-local random variables $X_1,\ldots,X_n$ such that for every $S \subseteq [n]$ with $|S| \leq k$ we have $\E X^S = c_{|S|}$.
  This follows by letting the local distribution on $S \subseteq [n]$ be given by $\mu$.

  Now let $\mu$ be the uniform distribution on $\{ x \in \{ \pm 1\}^k \, : \, \iprod{x,1} = 0 \}$.
  By symmetry, $\E_{x \sim \mu} x^S$ can depend only on $|S|$.
  Furthermore, for any $i \neq j \in [k]$ we have $\E x_i x_j = -1/(k-1)$, since
  \[
  0 = \E \iprod{x,1}^2 = \sum_{ij \in [k]} \E x_i x_j = k + \sum_{i \neq j} \E x_j x_j\mper
  \]
  Extending $\mu$ to a $k$-local distribution $X_1,\ldots,X_n$, we have
  \[
  \sum_{i \neq j} - \E X_i X_j \geq \Omega(n^2 / k)\mper\qedhere
  \]
\end{proof}

\section{Spectral partitioning into modestly non-expanding parts}\label{app:partition}

In this section, we prove an analog of the result of Arora et al. \cite{ABS} for partitioning graphs with high {\em threshold rank}.
The result of Arora et al. gives an algorithm for partitioning the graph into parts of expansion $\le \epsilon$ when $\rank_{1-\epsilon'}(G)$ is large, for $\epsilon,\epsilon'$ close to $0$.
However, their result does not give guarantees for graphs that have large $\tau$-threshold rank when $\tau$ is close to $0$ rather than $1$.
Later, \cite{Steurer10} uses very similar techniques to prove an analogous result for the small-$\tau$ parameter regime,
showing that if $\rank_{\epsilon}(G)$ is large, then one can obtain a partition of expansion $\le 1 - \epsilon'$.
We provide a proof here for completeness.
We provide the proof for the case of unweighted graphs $G$; the proof remains essentially unchanged if the graph has weights $w_{ij}$ such that $w_{ij} \geq w_{\max}/\text{poly}(n)$, which we can assume in the CSP context by throwing out low-weight edges.

\begin{theorem*}[Restatement of Theorem~\ref{thm:partition}]
Fix any $\tau,\alpha \in (0,1)$, and take $n$ sufficiently large.
Then there is a polynomial-time spectral algorithm that partitions an $n$-vertex simple graph $G =(V,E)$ into components $G_1,\ldots,G_m$ such that for all $i \in [m]$, the threshold rank $\rank_{\tau}(G_i) \le n^{\alpha}$, with the total fraction of edges cut in the partition bounded by $\frac{1}{|E|}|\cup_{i\neq j\in[\ell]} E[G_i,G_j]| \le 1 - \exp\left(-O\left(\frac{1}{\alpha^2} \log \frac{1}{\tau}\right)\right)$.
\end{theorem*}

We will prove this theorem by showing that if $\rank_\tau(G) \ge n^{\alpha}$, then we can in polynomial time find a subset of vertices that account for a $\approx n^{-\alpha/2}$ fraction of the volume with expansion $\le 1-f(\tau)$ for $f$ a function going to $0$ with $\tau$.
Formally,

\begin{lemma}\label{lem:nibble}
Let $c > 0$, let  $k \ll n$ be integers, let $G$ be an $n$-vertex graph, and suppose $\rank_{\tau}(G) \ge k$.
Then in polynomial time we can find a subset of vertices $S$ with $\frac{\vol(S)}{\vol(G)} \le O\left(\frac{1}{\tau^{c +2}\sqrt{k}} \left(\frac{n^2}{k}\right)^{1/c}\right)$ such that $1_S^\top A 1_S \ge \eta 1_S^\top D 1_S$, for $\eta = \Omega\left(\tau^{2c + 4}\right)$.
\end{lemma}

We will prove Lemma \ref{lem:nibble}, then below show how to apply this lemma recursively to obtain Theorem \ref{thm:partition}.
\begin{proof}[Proof of Lemma \ref{lem:nibble}]
Throughout the proof, for a graph $G$, we will let $A$ denote the adjacency matrix, $D$ denote the diagonal degree matrix, $P = D^{-1}A$ denote the transition matrix of the random walk on $G$, and $N = D^{-\hf} A D^{-\hf}$ denote the symmetric normalized adjacency matrix.

First, we will establish the existence of a standard basis vector $e_i$ with a large projection into the large eigenspaces of $N$ (relative to the degree of the corresponding vertex $i$). 
This will be a ``good starting point'' for identifying the modestly non-expanding set.
From this point on, our proof is similar to \cite{ABS}: we will apply a sparse Cheeger's inequality to show that if one takes an $T = O(\log n)$ step random walk, then there is some step $t \le T$ such that a level set of the measure of the $t$-step random walk starting from $i$ has large volume and does not expand too much.

\paragraph{Existence of a good starting point.}
Let $\Pi$ be the projector to the eigenspace of $N$ with eigenvalues $\ge \tau$.
Since $\rank_\tau(G) \ge k$, we have that
\[
k \le \Tr(\Pi) = \sum_{i \in [n]} \|\Pi e_i\|^2.
\]
Similarly, let $v = \frac{1}{\sqrt{\Tr(D)}}D^{\hf} \vec{1}$, and note that $v$ is a unit vector. 
From this we have that
\[
1 = \sum_{i \in [n]} \iprod{v, e_i}^2.
\]
Taking these together, we must have at least one standard basis vector $e_i$ such that
\begin{align}
\frac{\|\Pi e_i\|^2}{\iprod{v,e_i}^2} \ge \frac{\sum_{j \in [n]} \|\Pi e_j\|^2}{\sum_{j \in [n]} \iprod{v,e_j}^2 } \ge k.\label{eq:starting-pt}
\end{align}
and furthermore we have that $\iprod{v,e_i}^2 = \frac{\deg(i)}{\Tr(D)} > 0$ since no vertex is isolated.
The inequality (\ref{eq:starting-pt}) makes vertex $i$ and the associated standard basis vector $e_i$ a ``good starting point.''
Without loss of generality, suppose that $e_1$ is a good starting point, and define define $\mu = \|\Pi e_1\|^2$.

\paragraph{Finding an analytically non-expanding set via random walk.}
Take $T = \left\lfloor \frac{\log \frac{1}{\mu}}{c\log \frac{1}{\tau}}\right\rfloor$.
Since the eigenvalues of $N$ in $\Pi$ are at least $\tau$, we have a lower bound on the $\ell_2$ mass of the measure of the random walk\footnote{This is only precise for regular graphs. The measure of the random walk starting at $1$ is $\|e_1^\top P^t\|^2$. For regular graphs, $P = N$; in a non-regular graph, this is not quite precise.} starting from $1$:
\[
e_1^\top N^{2T} e_1 = \|N^T e_1\|_2^2 \ge \|\Pi e_1\|_2^2 \cdot \tau^{2T} = \mu \cdot \tau^{2T}.
\]
We also have that $N^0 = \Id$, and therefore $e_1^\top N^0 e_1 = 1$.
Taking a ``telescoping product,'' this implies that
\[
\mu\cdot  \tau^{2T} \le \frac{e_1 N^{2T} e_1 }{e_1^\top N^0 e_1} = \prod_{t = 1}^T \frac{e_1^\top N^{2t} e_1}{e_1^\top N^{2(t-1)}e_1}
\]
In particular, there must be some $t \in [T]$ which is at least as large as the geometric mean of these $T$ terms.
Let $t$ be that index, so that
\[
\frac{e_1^\top N^{2t} e_1 }{e_1^\top N^{2(t-1)}e_1} \ge \mu^{1/T} \tau^2 \ge \tau^{2+c}.
\]

For convenience, call $\tau' = \tau^{2+c}$, and call $z = N^{t-1} e_1$.
The above can be re-written as a statement about the Rayleigh quotient of $z$ in $N^2$,
\[
\frac{z^\top N^2 z}{z^\top z} \ge \tau'.
\]
We will take the vector $y = z + \frac{1}{2} N z$, and use the fact that $N$, $z$ have non-negative entries to show that $y$ has a large Rayleigh quotient with $N$.
From the triangle inequality, $\|y\|^2 \le \left(\|z\| + \frac{1}{2} \|N\| \|z\|\right)^2 \le \frac{9}{4} \|z\|^2$, and from the non-negativity of $N,z$, $y^\top N y = z^\top N z + z^\top N^2 z + \frac{1}{4} z^\top N^3 z \ge z^\top N^2 z$.
Therefore,
\[
\frac{y^\top N y}{y^\top y} \ge \frac{4}{9} \tau'.
\]

Now, we perform one final manipulation, letting $u = D^{-\hf} y$.
Re-writing the above,
\[
\frac{u^\top A u}{u^\top D u} \ge \frac{4}{9}\tau'.
\]
We thus have a vector $u$ which is an analytically non-expanding set.

\paragraph{Rounding to a discrete set.}
We now apply a sparse Cheeger inequality to $u$ to obtain our set.
We will use the following lemma, which is a modification of \cite[Lemma 3.4]{ABS} to non-regular graphs which may also be found in \cite[Lemma 2.2]{Steurer-thesis}.
In Appendix \ref{app:spec} we provide a proof for completeness.
\begin{lemma}[sparse Cheeger for irregular graphs]\label{lem:local-cheeg}
Let $A,D$ be the adjacency and degree matrices of a simple $n$-vertex graph $G = (V,E)$.
Suppose $v \in \R^n$ is a non-negative vector such that $v^\top A v \ge \epsilon v^\top D v$, and $\|Dv\|_1 \le \theta v^\top D v$.
Then for each $\delta>0$ there is a distribution over vectors $x \in \{0,1\}^n$ with $x^\top D x \le \frac{1}{\delta} \theta v^\top D v$ such that
\[
\frac{x^\top A x}{x^\top D x} \ge \frac{1}{2}\epsilon^2 - \delta^2 \left(\frac{\Tr(D)}{v^T D v} + 1\right).
\]
\end{lemma}

To apply the lemma, we must bound the magnitude of $\|D u\|_1$.
Using the relations $u = D^{-\hf} y$, $y = z + \tfrac{1}{2} N z$, $z = N^{t-1}e_1$, and $D^{-\hf} N = P D^{-\hf}$ (where $P$ is the transition matrix), we obtain:
\begin{align}
Du 
= D D^{-\hf} \left(N^{t-1} + \frac{1}{2}N^t\right)e_1
= D P^{t-1} D^{-\hf} e_1 + \frac{1}{2} D P^t D^{-\hf} e_1. \label{eq:transition}
\end{align}
Now, for a non-negative vector $w$, $\|DPw\|_1 = \|Dw\|_1$, since $1^\top D = 1^\top A$, and $\|DPw\|_1 = 1^\top D P w = 1^\top A w = 1^\top D w$.
Applying this $t-1$ and $t$ times to (\ref{eq:transition}) above, we have
\[
\|Du\|_1 \le \frac{3}{2} \|D^{\hf} e_1\|_1 = \frac{3}{2} \sqrt{\Tr(D)} \iprod{v,e_1},
\]
where we recall that $v = \frac{1}{\sqrt{\Tr(D)}} D^{\hf} \vec{1}$.
So, $\|Du\|_1 \le \theta u^\top D u$ for $\theta = 6 \sqrt{\Tr(D)} \iprod{v,e_1}\frac{1}{u^\top D u}$. 
Applying the Lemma \ref{lem:local-cheeg} with $\delta = \frac{1}{9}\tau' \sqrt{\frac{u^\top D u}{\Tr(D)}}$, we get a set indicator vector $x$ with large Rayleigh quotient,
\[
\frac{x^\top A x}{x^\top D x} \ge \frac{4}{81}(\tau')^2 - \delta^2\left(\frac{\Tr(D)}{u^\top D u} + 1\right) \ge \frac{2}{81} (\tau')^2.
\]
We also obtain the following bound on the relative volume,
\[
\frac{x^\top D x}{\Tr(D)} 
\le \frac{1}{\Tr(D)} \frac{1}{\delta} \theta u^\top D u 
\le \frac{1}{\Tr(D)} \cdot \frac{9}{\tau'} \sqrt{\frac{\Tr(D)}{u^\top D u}} \cdot 6 \sqrt{\Tr(D) \iprod{v,e_1}^2}
= \frac{54}{\tau'}\sqrt{\frac{\iprod{v,e_1}^2}{u^\top D u}}.
\]
Which we can simplify by noting that
\[
u^\top D u = \|y\|^2 \ge \frac{1}{4}\|z\|^2 = \frac{1}{4} e_1^\top N^{2t -2} e_1 \ge \frac{1}{4}\|\Pi e_1\|^2 \tau^{2T} \ge \frac{1}{4}\mu^{1 + 2/c}.
\]
Combining the above,
\[
\frac{x^\top D x}{\Tr(D)} 
\le \frac{54}{\tau'}\sqrt{4\frac{\iprod{v,e_1}^2}{\mu^{1 + 2/c}}}
\le \frac{108}{\tau'}\frac{1}{\sqrt{k}} \mu^{-1/c},
\]
where we have used the ``good starting point conditon''  $\mu/\iprod{v,e_1}^2 \ge k$ guaranteed by (\ref{eq:starting-pt}).
Finally, we use that $\iprod{v,e_1}^2 = \frac{d_1}{\Tr(D)} \le \frac{\mu}{k}$ to get that $\mu \ge \frac{k}{n^2}$, from which we have our conclusion.
\end{proof}

Now, we prove Theorem \ref{thm:partition}.
\begin{proof}[Proof of Theorem \ref{thm:partition}]
Let $k = n^{\alpha}$, and notice that by assumption $\vol(G) \le n^2$.
Let $r = \frac{12}{\alpha} \ge \frac{8\log n}{\log k} - 4$, so that $\left(\frac{n^2}{k}\right)^{1/r} \le k^{1/4}$.
Call $\tau ' := \tau^{r + 2}$; notice that by assumption $\tau'$ is a fixed constant independent of $n$.
Let $\ell = \frac{8}{\alpha} \ge 4\frac{\log \vol(G) - \log k}{\log k - \frac{1}{4}\log \frac{1}{\tau'}}$.
We apply Lemma \ref{lem:nibble} with $c = r$, $k = n^{\alpha}$, and $\tau$ recursively to obtain a partition $P$: we begin with $P_0 = \{G\}$, and then at every time step $t \ge 0$, if there exists a component $C$ in $P_t$ with $\rank_{\tau}(P_t) \ge k$, we apply Lemma \ref{lem:nibble} to remove from $C$ a subgraph $C'$ with $\frac{\vol(C')}{\vol(C)} \le O(\frac{1}{\tau' k^{1/4}})$ and $\phi_C(C') \le 1 - O((\tau')^2)$, then we take $P_{t+1}$ to be the refinement of $P_t$ given by splitting $C'$ off of $C$, $P_{t+1} = (P_t \setminus \{C\}) \cup \{C', C \setminus C'\}$.
This process must terminate, since a component of volume $\le k$ has rank at most $k$.

Let $P$ be the final partition, and consider any part $G_i$ in $P$.
Suppose that $G_i$ was obtained by subdividing $C_0 = G$, then $C_1,\ldots,C_m$.
Since at each subdivision the volume dropped by a factor of $\tau'^{-1}k^{-1/4}$, there cannot have been more that $\ell$ subdivisions that created $G_i$.
Further, at each subdivision the components kept a $O((\tau')^2)$ fraction of their internal edges.
Thus, $\phi_G(G_i) \le 1 - \left(c_1\tau'\right)^{2\ell}  = 1 - \exp\left(-O(\frac{1}{\alpha^2}\log\left(\frac{1}{\tau}\right))\right)$.
\end{proof}

\section{Global Correlation Rounding}\label{app:gcr}

First, we will prove that the global mutual information drops in expectation under conditioning.
\begin{lemma}[\cite{RagT12}, Lemma 4.5]\label{lem:potential}
Let $\mu$ be a distribution over $[n]$, and let $X_1,\ldots,X_n$ be $q$-ary random variables.
Then for all $k \in \N$ there exists $t \le k$ such that 
\[
\E_{i_1,\ldots,i_t \sim \mu} \E_{a,b \sim \mu}[I(X_a;X_b | X_{i_1},\ldots, X_{i_t})] \le \frac{\log q}{k}.
\]
Furthermore, this is true if $X_1,\ldots,X_n$ are $(k+2)$-local.
\end{lemma}
\begin{proof}
We use the identity $H(A|B,C) = H(A|C) - I(A;B|C)$ to obtain that for any $i_1,\ldots,i_{\ell}$,
\[
H(X_a| X_b, X_{i_1},\ldots,X_{i_\ell}) = H(X_a | X_{i_1},\ldots, X_{i_{\ell}}) - I(X_a;X_{b} | X_{i_1},\ldots,X_{i_{\ell}}).
\]
Rearranging and summing this equality over $\ell \in [k]$,
\[
\sum_{\ell = 1}^k I(X_a;X_b | X_{i_1},\ldots,X_{i_\ell}) = \sum_{\ell = 1}^k H(X_a | X_{i_1},\ldots,X_{i_\ell}) - H(X_a| X_b, X_{i_1},\ldots,X_{i_\ell}).
\]
Now, taking an expectation over $a,b,i_1,\ldots,i_{\ell} \sim \mu$ chosen independently, we have a telescoping of the sum on the right-hand side,
\[
\E_{a,b,i_1,\ldots,i_{k} \sim \mu} \sum_{\ell = 1}^k I(X_a;X_b | X_{i_1},\ldots,X_{i_\ell}) = \E_{a, i_1 \sim \mu} H(X_a | X_{i_1}) - \E_{a,b,i_1,\ldots,i_{k} \sim \mu} H(X_a| X_b, X_{i_1},\ldots,X_{i_\ell}) \le \log q,
\]
where the bound on the right-hand-side follows from the fact that $H(X_a) \le \log q$ and conditioning decreases entropy.
The bound now follows by taking the smallest of the $k$ terms in the sum.
\end{proof}

Next, we show that the objective value does not change in expectation under conditioning.
\begin{fact}[See e.g. \cite{BRS11}] \label{fact:obj-value}
If $\pE$ is a degree-$(D + k)$ Sherali--Adams pseudodistribution on local random variables $X_1,\ldots,X_n$ and $S \subset [n]$ has size at most $k$, then conditioning on $X^S = x^S$ for $x^S$ sampled according to the joint marginals defined by $\pE$ yields a degree-$D$ pseudodistribution $\pE_x'$ such that $\pE[q(X)] = \E_{x^S}\pE_x'[q(X)]$ for all $q$ of degree at most $D$ in $X$.
\end{fact}
\begin{proof}
By definition, we have that 
\[
{\pE_x}'[q(X)] = \begin{cases}
\frac{\pE[q(X) \cdot \Ind(X^S = x^S)]}{\pE[\Ind(X^S = x^s)]} & \pE[\Ind(X^S = x^s)] > 0 \\
0 & \text{otherwise}.
\end{cases}
\]
Therefore, since we set $X^S = x^s$ with probability $\pE[\Ind(X^S = x^S)]$,
\[
\E_{x} {\pE_x}'[q(X)]
= \sum_{x^S} \frac{\pE[q(X) \cdot \Ind(X^S = x^S)]}{\pE[\Ind(X^S = x^s)]} \cdot \Pr[X^S = x^S]
= \pE[q(X)],
\]
as desired.
\end{proof}

\begin{proof}[Proof of Theorem \ref{thm:gcr-conditioning}] 
We first apply Pinsker's inequality to deduce that $\pC(X_i,X_j)^2 \le 2I(X_i;X_j)$.
The theorem now follows by taking Fact \ref{fact:obj-value} together with Lemma \ref{lem:potential} and applying Markov's inequality together with a union bound.
\end{proof}

Finally, the following lemma shows that when the local covariance is small in absolute value, independent rounding achieves a large objective value in expectation.
\begin{lemma}[See e.g. \cite{BRS11}]
Suppose that $(G,\Pi)$ is an instance of a 2CSP, and suppose that $X_1,\ldots,X_n$ are $2$-local variables supported on $[q]$.
If $\E_{i\sim j} |\pC(X_i,X_j)|\le \delta$, then in expectation independent rounding according to the marginals of the individual $\{X_i\}$ will obtain objective value $\ge \E_{i\sim j} \pE \Pi(X_i,X_j) - \delta$.
\end{lemma}
\begin{proof}
Let $Y_i \in [q]$ be the random variable defined by the marginal probability distribution $\{X_i\}$.
By definition $\pC(X_i,X_j) = \|\{X_i,X_j\} - \{X_i\}\{X_j\}\|_1$, where $\{X_i,X_j\}$ is the joint distribution and $\{X_i\}\{X_j\}$ is the product of the marginals.
Thus, the statement of the lemma is given by noticing that since $\Pi$ is a $0/1$ function, for every $i\sim j$, 
\[
\E_{Y_i \sim \{X_i\}, Y_j \sim \{X_j\}} \Pi(Y_i,Y_j)
\ge \left(\E_{\{X_i,X_j\}} \Pi(X_i,X_j)\right) - \|\{X_i,X_j\} - \{X_i\}\{X_j\}\|_1
\ge \pE \Pi(X_i,X_j) - \delta 
\]
and the conclusion follows.
\end{proof}
Lemma \ref{lem:gcr-rounding} follows by applying standard derandomization arguments.

\section{Proofs of spectral primitives}
\label{app:spec}

This lemma appears (with somewhat different language) in \cite[Lemma 2.2]{Steurer-thesis}.
We provide a proof here for completeness.
\begin{lemma*}[sparse Cheeger for irregular graphs, restatement of \ref{lem:local-cheeg}]
Let $A,D$ be the adjacency and degree matrices of a simple $n$-vertex graph $G = (V,E)$.
Suppose $v \in \R^n$ is a non-negative vector such that $v^\top A v \ge \epsilon v^\top D v$, and $\|Dv\|_1 \le \theta v^\top D v$.
Then for each $\delta>0$ there is a distribution over vectors $x \in \{0,1\}^n$ with $x^\top D x \le \frac{1}{\delta} \theta v^\top D v$ such that
\[
\frac{x^\top A x}{x^\top D x} \ge \frac{1}{2}\epsilon^2 - \delta^2 \left(\frac{\Tr(D)}{v^T D v} + 1\right).
\]
\end{lemma*}
\begin{proof}
For convenience, we choose $\varepsilon'$ to be the ratio $\varepsilon' := \frac{v^\top A v}{v^\top D v} \ge \epsilon$.
We perform a standard Cheeger dependent rounding scheme.
Choose $t$ so that $t^2$ is distributed uniformly in $[0,1]$, and set $x_i = \Ind[ v_i > t]$.
We have that
\[
\E[x_i^2] = \Pr(v_i^2 \ge t^2) = v_i^2.
\]
As a consequence, $\E[x^\top D x] = v^\top D v$.

We also have that
\[
\E[(x^\top D x)^2] = \sum_{i,j \in [n]}\E[ d_i d_j x_i^2 x_j^2] = \sum_{i,j} d_i d_j \cdot \Pr(\min(v_i^2,v_j^2) \ge t^2) \le \sum_{i,j} d_i d_j \cdot \Pr( v_i v_j \ge t^2) = \|Dv\|_1^2.
\]
Combining the above with Chebyshev's inequality,
\[
\Pr( x^\top D x \ge \frac{1}{\delta}\theta v^\top D v) \le \delta^2 \frac{\|Dv\|_1^2}{\theta^2 (v^\top D v)^2} \le \delta^2.
\]

Finally, to bound the quadratic form with $A$ we first consider the Laplacian.
We have that
\begin{align*}
\E x^T(D-A) x 
&= \sum_{i\sim j} \Pr[ \max(v_i^2,v_j^2) \ge t \ge \min(v_i^2,v_j^2)]\\
&= \sum_{i \sim j} |v_i^2 - v_j^2|
= \sum_{i \sim j} |v_i - v_j| \cdot |v_i + v_j| \\
&\le \sqrt{(v^\top (D-A)v \cdot v^\top(D+A)v}\\
&= \sqrt{(v^\top D v)^2 - (v^\top A v)^2}.
\end{align*}
And, since $x^T A x \le x^T D x$, we have that
\[
\E x^\top A x 
\ge v^\top D v\cdot \left(1 - \sqrt{\frac{(v^\top D v)^2 - (v^\top A v)^2}{(v^\top D v)^2}}\right)
= v^\top D v\cdot \left(1 - \sqrt{1 - (\varepsilon')^2}\right)
\ge v^\top D v \cdot\left(\frac{1}{2}(\varepsilon')^2\right)
\]

Now, we consider the distribution over $x$ conditioned on $\cE_\delta$, the event that $\|x\|^2 \le \frac{1}{\delta} \theta \|v\|^2$.
We have that 
\[
\E[ x^T A x | \cE_{\delta}] 
= \frac{1}{\Pr[\cE_{\delta}]} \left(\E[x^T A x] - \E[x^T A x ~|~ \overline{\cE}_{\delta}] \cdot \Pr[\cE_\delta]\right)
\ge \E[x^\top A x] - \delta^2 \cdot \Tr(D).
\]
We also have
\[
\E[x^T D x |\cE_{\delta}]
= \frac{1}{\Pr[\cE_{\delta}]} \left(\E[x^T D x] - \E[x^T D x ~|~ \overline{\cE}_{\delta}] \cdot \Pr[\cE_\delta]\right)
\le \frac{1}{1-\delta^2} \E[x^T D x].
\]

Putting this together, we have
\[
\frac{ \E[ x^T A x | \cE_{\delta}] }{ \E[x^T D x |\cE_{\delta}]}
\ge (1-\delta^2) \left(\frac{1}{2} (\varepsilon')^2 - \delta^2 \frac{\Tr(D)}{v^T D v}\right).
\]
Taking our distribution to be the distribution over $x$ conditioned on $\cE_\delta$, this gives our conclusion.
\end{proof}

\end{document}